\documentclass[]{amsart}
\usepackage{amsmath}
\usepackage{geometry}
\usepackage{graphics, float}
\usepackage{epstopdf} 
\usepackage[pdftex]{graphicx}
\usepackage{caption}
\usepackage{subcaption}
\usepackage{subfloat}
\usepackage[square,sort,comma,numbers]{natbib}
\usepackage{bm}
\usepackage{amssymb,latexsym}
\usepackage{amsthm}
\usepackage[utf8]{inputenc}
\usepackage[english]{babel}
\usepackage{booktabs}

\theoremstyle{definition}
\newtheorem{definition}{Definition}

\theoremstyle{assumption}
\newtheorem{theorem}{Theorem}

\theoremstyle{remark}
\newtheorem{remark}{Remark}

\newtheorem{assumption}{Assumption}
\newtheorem{lemma}{Lemma}
\newtheorem{condition}{Condition}

\usepackage{commath}
\theoremstyle{plain}
\usepackage{float}
\usepackage{bbm}
\usepackage{accents}
\usepackage{xcolor}
\usepackage{dsfont}
\numberwithin{equation}{section}
\usepackage{hyperref}

\usepackage{mathtools} 
\usepackage[Symbol]{upgreek}

\DeclarePairedDelimiter{\autobracket}{(}{)}
\newcommand{\br}[1]{\autobracket*{#1}}

\DeclarePairedDelimiter{\Prfences}{\lbrack}{\rbrack}
\newcommand{\pr}[1]{\Prfences*{#1}}

\DeclarePairedDelimiter\Brace{\lbrace}{\rbrace}
\newcommand{\bra}[1]{\Brace*{#1}}

\DeclarePairedDelimiter\nrm{\lVert}{\rVert}
\newcommand{\Rbb}{\mathbb{R}}
\newcommand{\Ebb}{\mathbb{E}}
\newcommand{\Pbb}{\mathbb{P}}
\newcommand{\cF}{\mathcal{F}}

\begin{document}
	\author{Chao Deng, Xizhi Su and Chao Zhou}	
	\date{\today}
	\title{Relative wealth concerns with partial information and heterogeneous priors}
	\maketitle

\begin{abstract}
We establish a Nash equilibrium in a market with $ N $ agents with the performance criteria of relative wealth level when the market return is unobservable. Each investor has a random prior belief on the return rate of the risky asset. The investors can be heterogeneous in both the mean and variance of the prior. By a separation result and a martingale argument, we show that the optimal investment strategy under a stochastic return rate model can be characterized by a fully-coupled linear FBSDE. Two sets of deep neural networks are used for the numerical computation to first find each investor's estimate of the mean return rate and then solve the FBSDEs. We establish the existence and uniqueness result for the class of FBSDEs with stochastic coefficients and solve the utility game under partial information using deep neural network function approximators. We demonstrate the efficiency and accuracy by a base-case comparison with the solution from the finite difference scheme in the linear case and apply the algorithm to the general case of nonlinear hidden variable process. Simulations of investment strategies show a herd effect that investors trade more aggressively under relativeness concerns. Statistical properties of the investment strategies and the portfolio performance, including the Sharpe ratios and the Variance Risk ratios (VRRs) are examed. We observe that the agent with the most accurate prior estimate is likely to lead the herd, and the effect of competition on heterogeneous agents varies more with market characteristics compared to the homogeneous case.\\




\end{abstract}

\setcounter{tocdepth}{1}
\tableofcontents


\section{Introduction}
This paper contributes to the theory of both portfolio optimization under partial information and the relative wealth criteria and of forward backward stochastic differential equations (FBSDEs for short). For the former, we establish a system of stochastic equations with the solution corresponding to the value function and the optimal control. The information is updated by a general filter, which could be nonlinear. We show the uniqueness of the solution to the fully coupled multi-dimensional FBSDE under certain assumption on the boundedness of the generator coefficients. We are the first to use the deep learning method to solve for the portfolio allocation strategy for a utility game, and explicitly exam the strategies under various of market conditions, investor's risk preferences and the informational hetegegeneity. In combination with martingale approach, the deep learning algorithm explores the structural features of the controlled process. On the theoretical aspect, we are the first to use the variational FBSDE in the multi-dimensional case to solve for the N-equation system that characterizes the Nash equilibrium for the utility game. This transformation motivates future applications to analyzes of coupled systems such as the mean field game under a general non-Markovian setting. Simulation of the optimal portfoliio and the wealth process reveals novel insights on the both the effect of information heterogeneity under the relativeness utility. The investors interact through competiton, and the investor with the most accurate information is likely to be the leader.

In practice, a fund can use the market average as a benchmark and measure its performance by how much it overperforms. This is the case of an investment $ \textit{with competition} $. We also refer to information incompeletenss as $ \textit{partial information} $. Our market consists of a risk-free bond and $ d $ stocks, $ S = (S^1, ..., S^d) $, for some integer $ d < \infty $. For simplicity, we assume the risk-free interest $ r = 0 $. The stock prices are continuous processes adapted to the filtration $ \mathcal{F}_t $ on a filtered probability space $ (\Omega, (\mathcal{F}_{t})_{t\leq T}, \mathbb{P}) $. Each stock $ S^i $ has a return rate that depends on the stock fundamentals, modeled by a hidden variable $ A_t \in C([0, T], \mathbb{R}^l)$. $ W $ and $ B $ are independent standard Brownian motions adapted to $ \mathcal{F} $, with $ W \in C([0, T], \mathbb{R}^d) $ and $ B \in C([0, T], \mathbb{R}^l) $ for an integer $ l < \infty $. The stock processes and hidden variable processes have the following dynamics: 
\begin{align}
\label{stock}
\frac{dS^i_t}{S^i_t} &= h^i(A_t)dt + \sum_{j=1}^{d}\sigma_w^{ij}dW^j_t + \sum_{j=1}^{l}\sigma_h^{ij}dB^j_t \quad (\text{observed}), \; \\
\label{hiddenprocess}
dA_t &= \mu(A_t)dt + m(A_t)dB_t \hspace{3cm} (\text{hidden}), \; 
\end{align} 
where the initial condition to the stochastic differential equation (SDE) \eqref{hiddenprocess}, denoted by $ A_0 $, is also unobserved and independent of the Brownian motions $ W $ and $ B $. The coefficients $ h(a) $, $ \mu(a) $ and $ m(a) $ are $ C^1 $ and Lipschitz continuous. These conditions ensure the existence and uniqueness of strong solution to the above SDEs. The dependence of stock returns on the hidden variable $ A_t $ is through the function $ h(a) $. We further assume the diffusion coefficient is positive definite and uniform elliptic.

 
The observable filtration is the one generated by the stock prices. We denote by $ \mathcal{F}^S_t $ for the $ \sigma $-algebra generated by $ (S_u)_{u \leq t} $. Clearly, $ \mathcal{F}^S \subset \mathcal{F} $. In the following discussion, we abbreviate $ h^i_t $ for $ h^i(A_t) $, and write $ \hat{h}^i_t = \Ebb\pr{h^i(A_t)|\mathcal{F}_t^S} $. Each investor is aware of the model of the hidden variable and the stock prices. Her initial belief on the distribution of return is a normally distributed random variable, which we denote by $ \hat{h}^i(A_0) \sim N(m^i, v^i)$ if investor $ i $'s initial belief is normal with mean $ m^i $ and variance $ v^i $.

To write stock prices under partial information in a complete market form, define the innovation process as 
\begin{equation*}
\nu^i_t = \int_{0}^{t}\br{\frac{dS^i_u}{S_u^i} - \hat{h}^i_u du} = \sigma^i \zeta^i_t,
\end{equation*}
where $ \zeta^i_t $ is an $ \mathcal{F}_t^S $-adapted standard Brownian motion. 

With the total variance,
\begin{equation*}
\sigma = \br{\sigma_w \sigma_w^\intercal + \sigma_h \sigma_h^\intercal}^{1/2}.
\end{equation*}
The stock dynamic (\ref{stock}) can be written as 
\begin{equation*}
\frac{dS_t}{S_t} = \hat{h}_t dt + \sigma d\zeta_t.
\end{equation*}

The objective of an individual investor or a fund manager is to find a portfolio allocation strategy over available assets such that it maximizes the expected utility, which depends on the wealth amount that exceeds the average of all investors at terminal time $ T $. The investors are of  CARA type, mathematically, the utility function is
\begin{equation}
U^i(X_T^i, \overline{X}_T) = -e^{-\frac{1}{\delta^i}\big(X_T^i - \theta^i \overline{X}_T\big)}, \quad \text{where} \quad \overline{X}_T = \frac{1}{N} \sum_{k = 1}^{N}X_T^k.
\end{equation}
The parameters $ \delta^i > 0 $ and $ \theta^i \in [0, 1] $ represent the $ i $-th agent’s absolute risk tolerance and competition weight. A high value of $ \delta $ implies a high risk tolerance which, in general, induces an aggressive investment strategy. The case $ \theta = 0 $ corresponds to an investor with no relative wealth concern. 

We adopt the convention for CARA utility to denote the dollar amount of investment by $ \pi_t \in \mathbb{R}^d $. We aim to identify a Nash equilibrium $ \uppi^* = (\pi^{i, *}_t, ... \pi^{N, *}_t)_{t \in [0, T]} $. The strategy is optimal in the sense that no one is better-off by unilaterally deviating from it. When the return process is linear Gaussian, we can derive the dynamics of estimated return rate, then use the PDE approach to solve the investor's problem. The value function depends on the Markovian state variables that consist of the investor's wealth and the estimated return rate. We then derive an HJB equation for the value function. For exponential utility, we can reduce the dimension of the PDE. The resulting PDE only depends on the spacial variable $ \hat{h}_t $. In other words, increasing the number of agents in the game does not increase the dimension of the problem. We obtain an analytical solution of the value function, hence the equilibrium strategy. The strategies of all investors can be solved through a linear system whose coefficients depend on risk preferences, observable market parameters, estimates of market returns, and the investment horizon.

BSDE is an essential tool for the problem of a single investor under partial information. In the case of a nonlinear hidden variable process, the mean return cannot be written as a deterministic function of any finite-dimensional state. Therefore, the control problem is non-Markovian. Using a non-standard martingale representation theorem, we can write an $ \mathcal{F}^S $- adapted martingale as a stochastic integral against the innovation process. We derived a non-Markovian one-dimensional BSDE similar to the one in \cite{hu2005utility}. Combining the one-dimensional BSDE derived for the single-agent problem, we obtain a multi-dimensional fully coupled FBSDE, by which the terminal condition for the unidimensional BSDE is endogeneously determined.
 
The Nash equilibrium is unique under certain assumptions on the market parameters and risk preferences. The uniqueness of equilibrium follows from the uniqueness of the FBSDE solution. Since the return parameter in the FBSDE comes from estimation, it is bounded if the investor has prior knowledge about the range of the true return process. Under this assumption, the generator $ f $ is linear and Lipschitz, so are the drift and volatility of the forward wealth process, which we denoted by $ g(t, \cdot) $ and $ \sigma(t, \cdot) $, respectively. We show the uniqueness of the solution for fully coupled FBSDE in this case. Furthermore, since $ \sigma(t, \cdot) $ can be degenerate, many of the existing results for the well-posedness of FBSDEs do not apply. However, observe an important feature of this FBSDE, that the forward equation does not depend on the $ Y $ component of the backward equation. We can apply the main theorem in \cite{zhang2006wellposedness} to establish the uniqueness and existence of the solution.

We numerically solve this multi-dimensional FBSDE by a deep learning method. The numerical scheme is conducted in two stages. In Stage I, we estimate the stock returns using an $ L^2 $ projection. The recurrent neural network (RNN) is used in order to exploit the time series feature of the input to facilitate the sequential learning. The RNN first produces the hidden state, which will be transformed by a linear map to the final output corresponding to the estimation at each discrete time step. The RNN as a function approximator takes the stock paths from time $ 0 $ up to time $ T $, as well as the investor's initial belief as the input. However, the estimation at each time step depends only on the past stock prices and the investor's initial belief. The estimated return process appears in the drift and diffusion terms of the forward equation and the generator of the backward equation. In Stage II, we solve the FBSDE, again using neural networks as function approximators. The algorithm in this step is similar to \cite{weinan2017deep}. The FBSDE coupling requires no extra care in designing the neural network structure. However, the loss function must include a terminal condition that is a function of the wealth process, which is also computed from the NN parameters. We denote by the terminal loss for the difference between the parametrized function and the forward simulation at the terminal time. Experiments of the deep learning scheme on different sets of model parameters show the efficiency and robustness of our method. The deep learning method is flexible in that it allows the nonlinear type of filters. Moreover, deep learning can be easily adapted to multi-asset cases where most numerical scheme fails due to the explosion of the number of grids as the problem dimension grows, also known as the curse of dimensionality.

Investment strategies are compared through time series statistics. With our choice of market and risk parameters, in the linear Gaussian case, the standard deviations of the absolute value of the investment strategy are larger when with full information. Competition increases the mean and the standard deviation. We also compute the coefficient of variation (CV) as the ratio between the Std and the Mean. On average of three agents, competition does not change the CV significantly, which means the increase of the volatility of the investment strategy mainly attributes to the increase of the value itself but not the variation cross time. The Sharpe ratio and VRR further illustrates the performance of the optimal strategies for the utility game. In case of nonlinear filters, we found that the CV is significantly smaller under partial information for the three sets market parameters, indicating the strategies is less volatile when the investors estimate returns from the market. Competition may reduce the average CV. The numerical results are presented in section \ref{numerical_results}.

Comparing the case with linear and nonlinear filter, we observe that the agent with the most accurate prior estimate is likely to lead the herd, and the effect of competition on heterogeneous agents varies more with market characteristics. More generally, agent heterogeneity exploits market properties. This finding provides extra reasons why we introduce the information heterogeneity with agents' interactions into the portfolio model.

\subsection{Literature review}
Competition among fund managers stemmed from various incentives, from career advances motives to purpose of seeking clients. The empirical works \cite{brown1996tournaments}, \cite{agarwal2009role} and \cite{brown2001careers} have documented the phenomenoen. However, competitions among a group of more than two managers have hardly be considered. Studies have also shown the importance of relative concerns in financial economics, including \cite{abel1990asset}, where the utility is one of the relative consumption levels. With a time-separable utility function, \cite{gomez2007impact} shows that the Joneses behavior yields portfolio bias when the agents face non-diversifiable risks. \cite{demarzo2008relative} examed empirical implications of the relative wealth concerns, providing an explanation to the finanical bubble. The work \cite{bielagk2017equilibrium} analyzed the effect of social interactions when a market derivative is traded to share the risk among investors with relative performance concerns. The above works yield the herd effect by risk sharing motives or relative utilities. A recent work that includes the private information into the model is \cite{qiu2017equilibrium}. It analyzed the informative trading of managers and the implications on the efficiency of asset prices using an one-period mean-variance criteria. The focus is on the price informativeness with different information structures. Our focus is on the dynamic investment strategy with information heterogeneous agents and therefore different from it.

The optimal investment problem was initiated by Merton in 1970s \cite{merton1975optimum} as part of the asset pricing theory. A mass body of literatures have been developed since then. Classicial works include \cite{pliska1986stochastic}, \cite{karatzas1987optimal}, and \cite{cox1989optimal}, among others. Duality approach was developed in \cite{kramkov1999asymptotic} and \cite{rogers2003duality} and has become a useful tool to solve incomplete market model, with \cite{papanicolaou2019backward} a recent application with partial information. Problems with portfolio constraints was in \cite{cvitanic1992convex} and \cite{zariphopoulou1994consumption}, among others. For other models with market frictions, the transaction cost case was considered in \cite{magill1976portfolio}, \cite{davis1990portfolio}, and \cite{ShreveSoner1994}. \cite{dai2011illiquidity} first studied portfolio problem with both transaction costs and position limits. 

The paper \cite{lacker2019mean} used the PDE approach to solve an N-agent game. They established a Nash equilibrium under which all agents maximize their utilities with relative performance concerns for a varity of standard time-separable utility functions. The HJB equation can be derived when the agents are of CARA or the CRRA type. However, the PDE method has its limitations. It is restricted to the case of the deterministic mean return process and may fail in our market model with partial information and information filtering. The idiosyncratic noise is specific for the individual stocks that is only available for a particular agent. In our model, the same set of stocks driven by the common noise are available to all agents. Our model can be modified without much difficulty to accomodate the case where the stocks are driven by the common noise, but agent $ i $ only invests in stocks $ i $ with return $ b^i $ and stock price $ S^i $.

Initial works on consumption-portfolio choice and asset pricing under partial information include \cite{gennotte1986optimal}. The separation principle holds in the linear Gaussian setting - the investor's optimal decision is equivalent to that of first estimate and then optimize. The HJB equation for the optimization step involves the estimated return as a new state variable. An early work \cite{detemple1986asset} derived the equilibrium asset price, while symmetric information was assumed. \cite{basak2005asset} built an asset pricing model with investors of heterogeneous beliefs, although the agents do not interact through the utility game as in our model.

Based on a martingale representation theorem, \cite{karatzas1991note} reduced the partial information portfolio optimization problem into a complete market problem. With linear Gaussian filtering, partial information was studied in \cite{brendle2006portfolio} where the loss of utility due to incomplete information was quantified. Partial information is also seen in \cite{lee2016pairs} for optimization of pair trade strategies, and \cite{wang2008kalman} for recursive optimization. 
\cite{pham2001optimal} used martingales and duality theory to the case of stochastic volatility, although without explicitly solving the optimal strategy. 

For information models, the nonlinear filter was less considered than the linear one. Wonham filter is used to estimate the states of a Markov chain. \cite{rieder2005portfolio} and \cite{sass2004optimizing} studied partial information with regime switching, with \cite{rieder2005portfolio} using the PDE approach and \cite{sass2004optimizing} using the martingale approach. The latter is in fact more general since it allows stochastic interest rates and multiple assets.

\cite{mania2010exponential} considered the exponential utiity under partial information in a general semi-martingale framework where the available information is part of the filtration generating the stock prices. It shows the equivalence to a new optimization problem formulated by the observable processes.
By reducing to complete market case, \cite{bjork2010optimal} solved explicitly optimal strategies for various utility functions under partial information. In addition to the Merton proportion, the strategy includes a hedging demand for the volatility of the return process. \cite{papanicolaou2019backward} used results from filtering, duality, and the BSDE theory to solve the investment problem, which includes the case of an unbounded mean return. It argues that the BSDE solution is the unique limit of solutions to a sequence of truncated problems with unique solutions obtained by a martingale representation theorem. In our model, extra difficulty arises due to the FBSDE coupling, and a uniqueness result for unbounded returns in the general setting could not be obtained similarly.

The work \cite{pardoux1992backward} established a correspondence between the path-dependent HJB equations and non-Markovian control problems. The technique of using BSDE to solve the portfolio maximization problem was first introduced by \cite{rouge2000pricing} and further developed by \cite{hu2005utility}. Both works consider the portfolio problem by indifference pricing for a contingent claim.  More recently, \cite{matoussi2015robust} applied BSDE approach to the robust utility maximization, and second-order BSDEs to the robust problem under volatility uncertainty. 

Compared to BSDEs, the theory of coupled FBSDE was developed more recently. Antonelli \cite{antonelli1993backward} first obtained the result on the solvability of an FBSDE over a “small” time duration. Later, the Four Step Scheme in \cite{ma1994solving} and the Method of Continuation in \cite{hu1995solution} and \cite{yong1997finding} were used to establish the well-posedness result on an arbitrary time duration. The main result used in this paper is from \cite{zhang2006wellposedness}, which covers the cases of the fully-coupled equation with a degenerate diffusion coefficient. Classical methods for solving the BSDEs includes the Monte Carolo method. More recently, \cite{weinan2017deep} introduced a numerical method using the deep neural networks (NNs). The neural networks approximate the conditional expectations after time-discretization. \cite{han2020convergence} proved the theoretical convergence of the deep learning (DL) algorithm for the coupled FBSDEs using properties of the discretized equations. A backward scheme that also treats the deep neural networks as function approximators was developed in \cite{Hur2020SomeML}, where theoretical convergence of the numerical scheme was also shown. A deep NN based method to solve the mean field game (MFG) is in \cite{carmona2019convergenceI} and \cite{carmona2019convergenceII}, respectively for the ergodic and the finite horizon case. A systematic study of the performance of the DL method for solving (F)BSDEs with varying hyperparameters and network structures is yet missing, much less is the convergence and the stability the optimization with stochastic gradient descent (SGD). In our paper, the solution to the FBSDE is unique, hence we regard the loss value as an indicator of the training accuracy, and the convergence is justified by the theoretical features of the FBSDE and the empirical success of the DL algorithm.

Using a martingale approach, \cite{hu2005utility} derived BSDEs for investors with portfolio constraints when return is stochastic and uniformly bounded. Due to the boundedness of coefficients, the existence and uniqueness result is classical. \cite{espinosa2015optimal} studied an investment problem with relative performance concerns and the argument relies on \cite{hu2005utility}. The paper identified the market average as the payoff of a contingent claim and solves it in the utility indifference pricing framework. The generator was quadratic due to the portfolio constraint. For deterministic mean returns, they were able to derive an analytical solution for the N-agent equilibrium and show the solution exists by verification. The uniqueness result for stochastic coefficients is not available since the FBSDE was fully coupled with quadratic generators. In our model, the equation is simplier with linear coefficients.

The following table summarizes the works and correponding methods mentioned above that are the most closely related to this paper.
\begin{table}[ht]
\caption{Comparison of works}
\resizebox{\textwidth}{!}
{\begin{tabular}{c c c c c}
\toprule
Paper & Hu, Imekeller, Muller & Espinosa, Touzi & Zariphopoulou, Lacker & This paper \\
\hline\hline 
Utility  & exponential, power, log  & general utilities & exponential, power, log & exponential \\
Game & No & Yes & Yes & Yes \\
Class of return process & Stochastic & Deterministic & Deterministic & Stochastic \\
Main method & BSDE & FBSDE & PDE & FBSDE \\
Analytical solution & No & Yes & Yes & No \\
Adapt to portfolio constraints & Yes & Yes & No & Yes \\
Numerical solution & No & Yes & Yes & Yes \\
Learning & No & No & No & Yes \\
\bottomrule
\end{tabular}}
\end{table}

\section{Market model and the control problem}
\subsection{Preliminaries on market model}
For simplicity and without loss of generality, we assume that risk-free interest rate $ r = 0 $. The price dynamics for stock $ i $ is
\begin{align}
\frac{dS_t^i}{S_t^i} &= h^i(A_t) dt + \sum_{j = 1}^d \sigma_w^{ij} dW_t^j + \sum_{j = 1}^l \sigma_h^{ij} dB_t^j,
\end{align} 
for $ i \in \{1, ..., d\} $. $ \sigma_w \in \Rbb^{d \times d} $ and $ \sigma_h \in \Rbb^{d \times l} $ are constants. The Browian motions $ W $ and $ B $ are $ \mathbb{R}^d $-valued and $ \mathbb{R}^l $-valued, respectively and independent of each other. The functions $ h^i $ are $ \mathbb{R} $-valued whose forms are to be specified.

Following the partial information model in \cite{papanicolaou2019backward}, we consider two classes of filters. Recall that the relation between asset return rate and the hidden variable is specified by the function $ h(a) $. The following examples considers different dimensionality of the hidden variable $ A_t $. \\

$ \textbf{Example 1} $ (Multi-dimensional case). Suppose $ A_t $ is $ \mathbb{R}^{d^L} $-valued. $ h(A_t) $ is linear in $ A_t $. Suppose $ A_t $ follows an Ornstein-Uhlenbeck (OU) process that reverts to a constant $ \bar{\bm{\mu}} \in \mathbb{R}^{d^L}$:
\begin{equation}
dA_t = - \lambda \br{A_t - \bar{\bm{\mu}}} dt + \sigma_a dB_t,
\end{equation}
where $ \bar{\bm{\mu}} \in \mathbb{R}^{d^L} $, $ \sigma_a \in \mathbb{R}^{d^L \times l} $ and $ B $ is the $ l $-dimensional standard Brownian motion. The function $ h $ is linear with $ h(a) = (w^L)^T a + c^L $ for $ w^L $ in $ \mathbb{R}^{d^L} $ and $ c^L \in \mathbb{R} $. Based on the stock prices, investors update beliefs according to a Kalman filter (KF)\footnote{Refer to \cite{liptser2013statistics} for a general theory of Kalman-Bucy filter.}.

Our framework of FBSDE is general enough to include the class of nonlinear filters, as in the following examples:

\textbf{Example 2} (One-dimensional case).
\begin{align}
dA_t = -\lambda (A_t - \bar{\mu})dt + \sigma_a \sqrt{(A_t - a_l)(a_u - A_t) }dB_t.
\end{align} 
\begin{remark}
In the above dynamic, $ A_t $ is essentially bounded between $ a_l $ and $ a_u $. 
\end{remark} 

\textbf{Example 3} (One-dimensional case).
The hidden variable $ A_t $ follows the Cox-Ingersoll-Ross (CIR) process that is mean reverting to $ \bar{\mu} $:
\begin{equation}
dA_t = - \lambda (A_t - \bar{\mu}) dt + \sigma_a \sqrt{A_t} dB_t.
\end{equation}


For the following discussion, $ \abs{\cdot} $ denotes the Euclidean norm in $ \mathbb{R}^m $, $ m \in \mathbb{N} $. For $ p > 0 $, $ L^p $ denotes the set of $ \mathcal{F}_T $ measurable random variables $ F $ such that $ \Ebb[|F|^p] < \infty $. For $ k \in \mathbb{N} $, $ H^k(\mathbb{R}^d) $ deonotes the set of all $ \mathbb{R}^d $-valued stochastic processes $ \phi $ that are predictable with respect to $ \mathbb{F} $ and such that $ \Ebb[\int_{0}^{T}|{\phi}|^k] < \infty $. $ H^\infty (\mathbb{R}^d) $ is the set of all $ \mathbb{F} $-predictable $ \mathbb{R}^d $-valued processes that are $ \lambda  \bigotimes \mathbb{P} $-a.e. bounded on $ [0, T] \times \Omega $, where $ \lambda $ is Lebesgue measure on $ \mathbb{R} $. Denote $ \mathcal{E}(X) $ for the exponential martingale of $ X $. 

Recall that for $ 1 \leq j \leq d $, the process $ \pi^j_t $ is the dollar amount invested in stock $ j $ at time $ t $. The number of shares to hold for stock $ j $ is therefore $ \frac{\pi_t^j}{S_t^j} $. 



%

\begin{assumption} (Uniformly elliptic)
The total variance $ \sigma \sigma^\intercal $ is bounded, i.e,
\begin{equation}
\frac{1}{\epsilon} \leq \sigma \sigma^\intercal \leq \epsilon < \infty.
\end{equation}
for some positive definite matrix $ \epsilon $.
\end{assumption}

\noindent
\begin{condition}[Novikov]\label{novikov}
The process $ {h}_t $, $ t \in [0, T] $ satisfies the Novikov condition:
\begin{equation}
\Ebb\pr{e^{\frac{1}{2 \epsilon^h} \int_{0}^{T}||{h}_t||^2 dt} }  < \infty.
\end{equation}
for some small constant $ \epsilon^h $.
\end{condition}

\begin{remark}
By $ h_t = h(A_t) $, the Novikov condition is satisfied if $ h(a) $ is a square-root or power function and $ A_t $ is essentially bounded. By Jensen's inequality,
\begin{align*}
\mathbb{E}\pr{e^{\int_{0}^{T}\frac{1}{2\epsilon^h}\nrm{\hat{h}_t}^2 dt }}
\leq  \frac{1}{T} \int_{0}^{T}\mathbb{E} e^{\frac{T}{2\epsilon^h}\nrm{h_t}^2} dt. 
\end{align*}
Hence, given the right-hand side is finite, all moments of the estimated return process are bounded. 
\end{remark} 

\subsection{Objective function under relative wealth concerns}
We consider the market with $ N $ investors, each with the constant absolute risk aversion (CARA), or the exponential risk preference. Investors are concerned about their performances valued by the wealth relative to the average of market investors at a future time $ T $. The market average is modeled by a random payoff at the terminal time in the utility function. Denote by $ \overline{X}_t $ the average wealth of investors at time $ t $, and $ X^i_0 = x^i \in \mathbb{R} $ the $ i $-th investor's initial wealth. Since we refer the problem as an N-agent utility game, we use the word agent and investors interchangeably in this paper.

The relative utility function for the CARA investor $ i $ is
\begin{equation*}
U^i(X_T^i, \overline{X}_T) = -e^{-\frac{1}{\delta^i}\big(X_T^i - \theta^i \overline{X}_T\big)}, \quad \text{where} \quad \overline{X}_T = \frac{1}{N} \sum_{k = 1}^{N}X_T^k.
\end{equation*}
The objective function for an arbitrary investor with the risk parameter $ (\delta, \theta) $ and wealth at time $ t $, $ X_t $ is 
\begin{equation}
J(t, \pi_1, ..., \pi_N) = \mathbb{E}\pr{- e^{-\frac{1}{\delta}(X_T - \theta \overline{X}_T)}}.
\end{equation}
where $ \delta > 0 $ is the $ \textit{personal risk tolerance} $. $ \theta \in [0, 1] $ is the investor's $ \textit{competition weight} $ parameter. In the following discussion, the superscript $ i $ indicates variables for investor $ i $. 

Write $ \tilde{X}_t^i = \frac{1}{N}\sum_{j\neq i} X_t^j$. By simple algebra, the objective can be written as
\begin{align}
J(t, \pi_1, ..., \pi_N) &= \mathbb{E}\pr{- e^{-\frac{1}{\delta^i}\br{(1 - \frac{\theta^i}{N})X_T^i - \theta^i \tilde{X}_T^i}}} \\\nonumber
& = \mathbb{E}\pr{- e^{-\frac{1}{\delta^i}\br{1 - \frac{\theta^i}{N}}\br{X_T^i - \frac{N \theta^i}{N - \theta^i}\tilde{X}_T^i}} }
\end{align}

The value function for agent $ i $ with initial wealth $ x^i $ is
\begin{equation}
V(0, x^i) = \underset{\pi_i \in \mathcal{A}^i}{\sup}\mathbb{E}\pr{- e^{-\frac{1}{\delta^i}\br{1 - \frac{\theta^i }{N}}\br{X_T^i - \frac{N \theta^i }{N - \theta^i }\tilde{X}_T^i}} }.
\end{equation}

\noindent
$ \textbf{Admissible set}. $
The set of admissible strategies $ \mathcal{A}^i $ for agent $ i $ is the set of all predictable processes $ \pi = (\pi_t)_{0 \leq t \leq T} $ such that (1) $ \Ebb\pr{\int_{0}^{T}\abs{\pi_t \sigma}^2 } < \infty $. (2) The set
\begin{equation*}
\bra{e^{\pm X_\tau^{i, \pi}}: \tau \text{ is a stopping time with value in } [0, T]}
\end{equation*}
is uniformly bounded in $ {L}^q(\mathbb{P})$ for all $ q > 0 $.

We define the Nash equilibrium as follows:
\begin{definition}[Nash equilibrium]
A vector $ (\pi^{1, *}, ..., \pi^{N,*}) $ of admissible strategies is a Nash equilibrium if, for all $ \pi^i \in \mathcal{A}^i $, $ i \in \{1, ..., N\} $, 
\begin{equation*}
J_i(\pi^{1,*}, ..., \pi^{i,*}, ..., \pi^{N,*} ) \geq J_i(\pi^{1,*}, ..., \pi^i, ..., \pi^{N,*} ).
\end{equation*}
\end{definition} 

\vspace{0.5em}
\section{Nash equilibrium by FBSDE}
For presentation simplicity, we set $ d = 1 $ for the rest of the discussion. The analysis would adapt to the multiple common stocks case with purely notational change. See remark \ref{rmk_multi_stock} for more details.
\subsection{Formal derivation of FBSDE}
Recall that when the terminal condition and $ b_t \coloneqq \frac{h_t}{\sigma} $ are bounded,  \cite{hu2005utility} has shown that the investor's value function and optimal strategy (without constraint) correspond to the solution $ (Y_t, Z_t) $ of the following BSDE:
\begin{equation}\label{udbsde}
Y_t = F - \int_{t}^{T}Z_sdW_s - \int_{t}^{T}f(s, Z_s)ds, \quad t \in [0, T],
\end{equation}
with the generator
\begin{equation}\label{udgen}
f(\cdot, z) = zb_t + \frac{\delta}{2 }|b_t|^2.
\end{equation}
and some bounded random variable $ F $ that is the terminal condition of the BSDE.

The optimal strategy is given by 
\begin{equation*}
\pi_t^* = \frac{p^*_t}{\sigma},
\end{equation*}
where $ p_t^* $ is linear in the $ Z_t $ component in the BSDE solution,
\begin{equation}
p_t^* = Z_t + \delta b_t, \quad t \in [0, T].
\end{equation}

The value function at the initial time is given by
\begin{equation}
V(x) = -e^{-\frac{1}{\delta} (x - Y_0)}.
\end{equation}

We now formulate the optimization problem under the relative performance criteria. The random variable $ F $ represents the benchmark market average at terminal time. The constant $ \alpha^i $ captures the risk preference and competition concern of the $ i $-th investor.  More explicitly, take $ F = \frac{N \theta}{N - \theta}\tilde{X}_T $, and $ \frac{1}{\delta} = \frac{1}{\delta^i}\big(1 - \frac{\theta^i}{N}\big) $ in \eqref{udbsde} and \eqref{udgen} , we obtain the objective function for the $ i $-th investor.

Suppose all agents adopt the strategies corresponding to $ p^* $, we can write explicitly the wealth process of agent $ i $ in terms of $ Z_t^i $ and estimated market parameters, 
\begin{align}
\label{xprocess}
X_t^i &= x^i + \int_{0}^{t} \sigma^{-1} p^{i,*}_u \frac{dS_u}{S_u} \\\nonumber
& = x^i + \int_{0}^{t} \br{Z_u^i + \delta b_u} \br{\frac{\hat{h}_u}{\sigma_u} du + d\zeta_u}. \nonumber
\end{align}

Since $ X^i $ may be unbounded, the terminal condition $ F $ does not satisfy the condition in \cite{hu2005utility}. Therefore, the correspondence between our optimization problem and the above BSDE is not immediate. We will show that under certain integrable conditions on the return rate, the solution to a single agent's investment problem is still characterized by the BSDE \eqref{udbsde}.

To write the system of BSDEs corresponding to each investor as a multi-dimensional equation, we introduce the vector notation:
\begin{align}
\bm{X} = \br{X^i}_{i \in \{1, ..., N\}}
\end{align} 
where $ X^i $ is for an arbitrary random variable that corresponds to the investor $ i $.

We introduce a matrix notation to compute the performance benchmark of the wealth amount.
\begin{align}
\bm{F} = A \bm{X}_T,
\end{align}
where $ A $ is an N-by-N matrix that does both averaging and taking into account the relative concern of a particular agent. To be more specific,
\begin{equation}\label{matrix_A}
A = 
\begin{bmatrix}
0 & \frac{\theta_1}{N - \theta_1} & \ldots & \frac{\theta_1}{N - \theta_1} \\
\frac{\theta_2}{N - \theta_2} & 0 & \ldots & \frac{\theta_2}{N - \theta_2} \\
\vdots & \vdots & \ddots & \vdots\\
\frac{\theta_N}{N - \theta_N} & \frac{\theta_N}{N - \theta_N} & \ldots & 0 
\end{bmatrix}.
\end{equation}

From now on, let $ f $ operates on $ \bm{z} $ componentwise, and let $ \circ $ denote the componentwise multiplication. If all agents solve the optimization problem, we can write
\begin{align}
f(\cdot, \bm{z}) = \bm{z} \circ \bm{b}_t + \frac{\bm{\delta}}{2} |\bm{b}_t|^2.
\end{align}
Notice that if $ \bm{b}_t $ is bounded, then $ f $ is Lipschitz. 

The above derivation suggests that any Nash equilibrium strategy for $ N $ agents with relative performance corresponds to the following multi-dimensional FBSDE:
\begin{align}
\label{mainforward}
\bm{X}_t & = \bm{x}_0 + \int_{0}^{t} \br{\bm{Z}_s \circ \bm{b}_s + \bm{\delta} \abs{\bm{b}_s}^2}ds + \int_{0}^{t} \br{\bm{Z}_s + \bm{\delta} \bm{b}_s} \circ d\bm{\zeta}_s, \\
\label{mainbsde}
\bm{Y}_t &= A \bm{X}_T - \int_{t}^{T}\bm{Z}_s \circ d\bm{\zeta}_s - \int_{t}^{T}f(s, \bm{Z}_s) ds,
\end{align}
where $ \bm{Z}_s \circ d \bm{\zeta}_s = \big(Z_s^i d\zeta^i_s \big)^\intercal_{i\in \{1, ..., N\}}$.

\begin{remark}\label{rmk_multi_stock}
As was previously mentioned, in the present setting, agents invest in the same set of stocks. Hence all components of $ \bm{b} $ and $ \bm{\zeta} $ are identical. In case agent $ i $ invests in stock $ i $ while all the stocks are still driven by the common noise, the components of $ \bm{b} $ and the Brownian motion $ \bm{\zeta} $ will depend on the agent $ i $'s individual stock.
\end{remark} 

\begin{theorem}
Suppose the return rate satisfies the Novikov condition. Let $ \hat{h}^i_t = \hat{h}(A_t) $ be the estimated mean return of the $ i $-th agent. If there exists a solution $ (Y_t, Z_t) $ to the FBSDE (\ref{mainforward} - \ref{mainbsde}), then there is a Nash equilibrium strategy  $ \uppi^* =(\pi^{1,*}, ..., \pi^{N,*}) $ given by
\begin{equation}
\uppi^* = \sigma^{-1}\br{\bm{Z}_t + \frac{\bm{\delta}}{\bm{b}_t}}. 
\end{equation}
\end{theorem}

In order to construct the investment strategy by the martingale approach under the observable filtration, we need the following lemma.
\begin{lemma}\label{lemmafiltration}
	Every martingale in $ \mathcal{F}^\zeta $ is a martingale in $ \mathcal{F}^S $.
\end{lemma}

\begin{proof}
	Let $ M $ be a $ \mathcal{F}^\zeta $-martingale. By a martingale representation theorem in \cite{bjork2010optimal}, 
	the martingale $ M_t^S = \mathbb{E}\pr{M_T | \mathcal{F}^S_t }$ has a unique representation:
	\begin{equation}
	M_t^S = \mathbb{E}\pr{M_T } + \int_{0}^{t}M_u^S \sigma^M_u d\zeta_u.
	\end{equation}
	Since $ M_t^S $ is also a stochastic integral against $ \zeta_t $, it is a martingale under $ \mathcal{F}^\zeta_t $. Alternatively, we can write it as  
	\begin{equation}
	M_t^S = \mathbb{E}\pr{M_T | \mathcal{F}^\zeta_t} = M_t.
	\end{equation}
	Since $ M^S $ is a $ \mathcal{F}^S $-martingale, so is $ M $.
\end{proof}

We next prove that the solution to the FBSDE we constructed is indeed a Nash equilibrium. The above lemma is used to show the martingale we will later construct is in a larger filtration $ \mathcal{F}^S $ than $ \mathcal{F}^\zeta $. From now on, we omit the super(sub)-script $ i $ when there is no ambiguity. 

\begin{proof}
We first show the solution to the unidimensional BSDE is the value function for a single agent's investment problem with terminal payoff $ F = \frac{N \theta}{N - \theta} \tilde{X}_T $. 

For investor $ i $, following the proof in \cite{hu2005utility}, first define a strategy $ p \in \mathcal{A} $ by 
\begin{equation}
R_t^{(p)} := -e^{- \frac{1}{\delta} (X_t^{(p)} - Y_t)}, \quad t \in [0, T],
\end{equation}
where $ Y_t $ is a component of the solution to the unidimensional BSDE 
\begin{equation}
Y_t = F - \int_{t}^{T}Z_sd \zeta_s - \int_{t}^{T}f(s, Z_s)ds, \quad t \in [0, T],
\end{equation}

Define
\begin{equation*}
M_t^{(p)} := - e^{-\frac{1}{\delta} (x - Y_0)} \mathcal{E}\br{- \frac{1}{\delta} \int_{0}^{t}\br{p_s - Z_s} d\zeta_s}.
\end{equation*}
which is a local martingale in $ \cF^\zeta $.


For
\begin{equation*}
C_t^{(p)} := e^{-\frac{1}{\delta} \int_{0}^{t}\br{b_s p_s - f(s, Z_s) - \frac{\delta}{2} \abs{p_s - Z_s}^2} ds}, \quad t \in [0, T],
\end{equation*}
we have 
\begin{equation*}
R_t^{(p)} = M_t^{(p)} C_t^{(p)}.
\end{equation*}

For $ R_t^{(p)} $ to be a local martingale for some $ p^* $ and a supermartingale for all $ p \in \mathcal{A} $, we need $ C_t^{(p)} $ decreasing and $ C_t^{(p^*)} = 1 $, $ \mu_L \bigotimes \mathbb{P} $-a.s. for $ \mu_L $ the Lebesgue measure on $ \mathbb{R} $ and some $ p^* $.

The exponent of $ C_t^{(p)} $ is a quadratic function in $ p $. Optimization yields,
\begin{equation}
f(\cdot, z) = z b_t + \frac{\delta}{2}\abs{b_t}^2,
\end{equation}
and 
\begin{equation}
p_t^* = Z_t + \delta b_t, \quad t \in [0, T].
\end{equation} 
Hence $ C_t^{(p^*)} = 1 $ and $ R_t^{(p^*)} = M_t^{(p^*)} $ is a local martingale in $ \mathcal{F}^\zeta $.  

Since $ b_t $ satisfies the Novikov condition, $ M_t^{(p^*)} $ and hence $ R_t^{(p^*)} $ are true martingales. The Novikov condition of $ b_t $ also implies that $ R_t^{(p^*)} $ is uniformly integrable. Hence, $ R_{\tau }^{(p^*)} $ is uniformly bounded in $ {L}^q(\Pbb) $, which implies that the strategy $ p^*_t $ is admissible.

It remains to show that $ R^{(p)} $ is a supermartingle for all $ p \in \mathcal{A} $ under $ \mathcal{F}^S $. Since the process $ M_t^{(p)} $ is a local martingale in $ \mathcal{F}^\zeta $, there exists a sequence of stopping times $ (\tau_n)_{n \in \mathbb{N}} $ such that $ \underset{n \to \infty}{\lim} \tau_n = T, \; \mathbb{P}$-a.s. and $ (M_{t \wedge \tau_n})_{t \in [0, T]} $ is a positive $ \mathcal{F}^\zeta $-martingale for each $ n \in \mathbb{N} $. By Lemma \ref{lemmafiltration}, they are $ \mathcal{F}^S $-martingales. The process $ \tilde{C}^{(p)} $ is decreasing. Thus $ R_{t \wedge \tau_n}^{(p)} $ is a $ \mathcal{F}^S $-supermartingale for each $ n $. That is, for $ s \leq t $,
\begin{equation*}
\mathbb{E}\pr{R_{t \wedge \tau_n}^{(p)}|\mathcal{F}^S_s} \leq R_{s \wedge \tau_n}^{(p)}.
\end{equation*} 

Equivalently, for any set $ U \in \mathcal{F}^S$, we have 
\begin{equation*}
\mathbb{E}\pr{R_{t \wedge \tau_n}^{(p)} \mathbbm{1}_U} \leq \mathbb{E}\pr{R_{s \wedge \tau_n}^{(p)} \mathbbm{1}_U}.
\end{equation*}

Notice that the admissible condition implies that $ e^{X_\tau} $ is uniformly bounded in $ {L}^q(\Pbb) $. By the forward equation for $ X $, we have
\[ e^{\int_{0}^{\tau} \br{Z_s b_s ds + Z_s d\zeta_s}} = e^{X_\tau - x_0 - \int_{0}^{\tau} \delta \abs{b_s}^2 ds - \int_{0}^{\tau} \delta b_s d\zeta_s} = e^{X_\tau} \cdot e^{- x_0 - \int_{0}^{\tau} \delta \abs{b_s}^2 ds - \int_{0}^{\tau} \delta b_s d\zeta_s} \]
is uniformly bounded in $ {L}^q(\Pbb) $. 

Furthermore, the solution for the backward equation $ Y $ has 
\begin{align*}
\abs{e^{Y_\tau}}^q &= \nrm{e^{Y_0 + \int_{0}^{\tau} Z_s d\zeta_s + \int_{0}^{\tau} \br{Z_s b_s + \frac{\delta}{2} \abs{b_s}^2}ds}}^q \\
& \leq e^{q C_0 Y_0} \cdot e^{q C_1 \int_{0}^{\tau} \br{Z_s d\zeta_s + Z_s b_s ds}} \cdot e^{q C_2 \int_{0}^{\tau} \frac{\delta}{2} \abs{b_s}^{2} ds} < \infty.
\end{align*}
Therefore, $ e^{Y_\tau} $ is uniformly bounded in $ {L}^q(\Pbb) $.


The process $ R_\tau^{(p)} $ is a constant multiple of the product of $ e^{- \frac{1}{\delta} X_\tau^{(p)}} $ and $ e^{\frac{1}{\delta} Y_\tau} $. Combining the definition of admissiblility and the above argument, we have 
\begin{equation}
\mathbb{E} [e^{q R_\tau^{(p)}}] \leq C \mathbb{E} [e^{- r \frac{1}{\delta} X_\tau^{(p)}}] \mathbb{E} [e^{r' \frac{1}{\delta} Y_\tau}] < \infty
\end{equation} 
where $ r = \frac{1 + q}{2} $ and $ \frac{1}{r} + \frac{1}{r'} = \frac{1}{q} $.

we have $ \bra{R_{t \wedge t_n}^{(p)}}_n $ and $ \bra{R_{s \wedge t_n}^{(p)}}_n $ are uniformly bounded in $ {L}^q(\Pbb) $ over $ n $ for all $ q > 0 $. Letting $ n \to \infty $, we obtain 
\begin{equation*}
\mathbb{E}\pr{R_t^{(p)}\mathbbm{1}_U} \leq \mathbb{E}\pr{R_s^{(p)}\mathbbm{1}_U},
\end{equation*}
which implies the supermartingale property of $ R^{(p)} $ as claimed. 
\end{proof}

\subsection{Explicit solution in case of linear filters}
In this section, we derive a PDE solution for linear Gaussian filters. Consider a model with the following stock dynamics, which corresponds to $ d = 1 $, $ h(a) = a $, $ \sigma_w = \sigma_S \sqrt{1 - \rho^2} $ and $ \sigma_h = \sigma_S \rho $, $ \mu(a) = -\lambda(a - \bar{\mu}) $, $ m(a) = \sigma_\mu $ in \eqref{stock} - \eqref{hiddenprocess}, 
\begin{align}\label{ou_stock}
\frac{dS_t}{S_t} &= A_t dt + \sigma_S \br{\sqrt{1 - \rho^2} dW_t + \rho dB_t}, \\
\label{ou_hidden}
dA_t &= - \lambda (A_t - \bar{\mu}) dt + \sigma_\mu dB_t, \quad t \in [0, T].
\end{align}
In the following discussion, we adopt the notation $ \hat{h}_t = \hat{h}(A_t) = \mathbb{E}[h(A_t)\big|\mathcal{F}_t^S \vee \hat{h}(A_0)] $. From filtering theory, the innovation process 
\begin{equation}
\nu_t = \int_{0}^{t} \br{\frac{dS_u}{S_u} - \hat{h}_u du}
\end{equation}
is a scaled Brownian motion under $ \mathcal{F}_t^S $, i.e, $ \zeta_t = \sigma^{-1} \nu(t) $ is a standard Brownian motion under $ \mathcal{F}_t^S $.

With the estimated return rate and the Brownian motion in a smaller filtration, the asset price can be regarded as one in the complete market. The price evolves as 
\begin{equation*}
\frac{dS_t}{S_t} = \hat{A}_tdt + \sigma_S d\zeta_t,
\end{equation*}
where the stochastic term is an $ \mathcal{F}_t^S $ adapted Brownian motion $ \zeta_t $.

We use the notation $ \Sigma(t) $ to indicate $ \Sigma $ is a deterministic function of $ t $. By the filtering theory, 
\begin{equation}
d\hat{A}_t = - \lambda ( \hat{A}_t - \bar{\mu}) dt +\br{\frac{\hat{\Sigma}(t) + \sigma_S \sigma_\mu \rho}{\sigma_S}} d\zeta_t,
\end{equation}
with $ \hat{A}_0 \sim N(\eta_0, \hat{\Sigma}(0)) $. In addition, the conditional variance $ \hat{\Sigma}(t) = \mathbb{E}\pr{(A_t - \hat{A}_t)^2 \big| \mathcal{F}_t^S \vee \hat{A}_0 } $ satisfies a Riccati ODE, with analytical solution 
\begin{equation}
\hat{\Sigma}(t) = \hat{\Sigma}(t; \Sigma_0) = \sqrt{k} \sigma_S \frac{k_1 e^{2 \frac{\sqrt{k}}{\sigma_S} t} + k_2}{k_1 e^{2 \frac{\sqrt{k}}{\sigma_S} t} - k_2} - \br{\lambda + \frac{\sigma_\mu \rho}{\sigma_S}} \sigma_S^2.
\end{equation}
The constants $ k $, $ k_1 $ and $ k_2 $ are
\begin{align*}
k &= \lambda^2 \sigma_S^2 + 2 \sigma_S \sigma_\mu \lambda \rho + \sigma_\mu^2, \\
k_1 &= \sqrt{k} \sigma_S + (\lambda \sigma_S^2 + \sigma_S \sigma_\mu \rho ) + \Sigma_0, \\
k_2 &= - \sqrt{k} \sigma_S + (\lambda \sigma_S^2 + \sigma_S \sigma_\mu \rho ) + \Sigma_0.
\end{align*}

Let $ (\alpha_t^i)_{t \in [0, T]}, i \in \{1, ..., N\} $ be an investment strategy among $ N $ investors. For the investor $ i $, denote the competitor's average by $ \alpha^{-i} $, where $ q^{-i}_t = \frac{1}{N} \sum_{j \neq i} q^j_t $ for arbitrary process $ q_t $.
\begin{align*}
dX^i_t &= \pi_t^i \hat{A}_t^i dt + \pi_t^i \sigma d\zeta_t,\\
d\tilde{X}^i_t &= (\alpha_t \hat{A}_t)^{-i} dt + \alpha^{-i}_t \sigma d\zeta_t.
\end{align*}

We can derive a PDE for $ V(t, x, y, \eta) = \underset{\pi \in \mathcal{A}}{\sup} \mathbb{E} \pr{J(t) \big | X_t = x, \tilde{X}_t = y, \hat{\mu}_t = \eta} $ and obtain the following results for the optimal strategy of the investors, from which we can compute the value function of each investor. We include the explicit expression for the value function in the appendix to save the space.

\begin{theorem}
The stock return rate satisfies the condition \ref{novikov}. $ \delta^i > 0 $, $ \theta^i \in [0, 1] $ for all $ i \in \{1, ..., N\} $, and define $ w_2^i = \frac{\delta^i}{\theta^i} $, the investor $ i $'s estimate of return is $ \eta^i $, $ i \in \{1, ..., N\} $. Then there exists a unique Nash equilibrium strategy among investors.

Define the following constants depending on $ i $,
\begin{multline*}
\mathfrak{m}^i = \frac{\delta^i}{\sigma_S^2 (1 - \frac{\theta^i}{N})} \Bigg\{\frac{1}{\sigma_S^2} w_2^i \br{\hat{\Sigma}(t) + \sigma_S \sigma_\mu\rho} \int_{t}^{T} e^{- \int_{t}^{s} \br{\lambda + \frac{\hat{\Sigma}(u) + \sigma_S \sigma_\mu\rho}{\sigma_S^2} }du} \\
 \cdot \int_{t}^{s}  \br{\hat{\Sigma}(u) + \sigma_S \sigma_\mu\rho} e^{- \int_{u}^{s} \br{\lambda + \frac{2(\hat{\Sigma}(m) + \sigma_S \sigma_\mu\rho)}{\sigma_S} - \frac{\hat{\Sigma}(m) + \sigma_S \sigma_\mu\rho}{\sigma_S^2}}dm} du ds + \frac{\theta^i}{1 - \frac{\theta^i}{N}}\Bigg\},
\end{multline*} 

and
\begin{multline*}
\beta^i = \frac{\delta^i}{\sigma_S^2 \br{1 - \frac{\theta^i}{N}}}\bigg\{\eta^i + \br{\hat{\Sigma}(t) + \sigma_S \sigma_\mu\rho} \int_{t}^{T} \br{- \frac{1}{\sigma_S^2}} e^{- 2 \int_{t}^{s} \br{\lambda + \frac{\hat{\Sigma}(u) + \sigma_S \sigma_\mu\rho}{\sigma_S}} du} \eta^i \\ + \frac{\hat{\Sigma}(u) + \sigma_S \sigma_\mu\rho}{\sigma_S^2} \int_{t}^{T} \int_{t}^{s} \lambda \bar{\mu} e^{- \int_{u}^{s} \br{\lambda + \frac{2(\hat{\Sigma}(m) + \sigma_S \sigma_\mu\rho)}{\sigma_S} - \frac{\hat{\Sigma}(m) + \sigma_S \sigma_\mu\rho}{\sigma_S^2}} dm} du  \\ \cdot e^{- \int_{t}^{s} \br{\lambda + \frac{\hat{\Sigma}(u) + \sigma_S \sigma_\mu\rho}{\sigma_S^2}} du} ds \bigg\}.
\end{multline*}
Let $ \mathcal{M} $ be the matrix with the $ i $-th row equals to $ \mathfrak{m}^i \bm{e}^i $, where $ (\bm{e}^i)_j = 1 - \delta_{\{i=j\}} $. The vector $ \upbeta $ with the i-th component $ \upbeta_i = \beta^i $. Then, the optimal strategy $ \uppi^* = (\pi^{*, i})_{i=1}^N $ can be expressed in terms $ \mathcal{M} $ and $ \upbeta $ as
\begin{equation*}
\uppi^* = \mathcal{M}^{-1} \upbeta.
\end{equation*}
\end{theorem} 

The detailed calculation is in the Appendix. Observe that the analytical solution includes all the market parameters and risk preferences of all investors. In our numerical experiments, we set the number of agents $ N = 3 $. Given $ \hat{\Sigma}(0) \in \mathbb{R} $, we compute $ \hat{\Sigma}(t) $ for $ t \in[0, T] $, then compute $ \mathcal{M} $ and $ \upbeta $ by numerical integration. Agents can have identical or heterogeneous initial belief on the hidden parameter.  At each time step $ t_k $, instead of solving a linear system that involves a matrix inversion, we iteratively solve for strategy $ \alpha^i_{t_k} $ until $ \alpha $ converges. The convergence is within $ 3 $ iterations at each time step, and hence the method is efficient. 


\section{Existence and Uniqueness of the FBSDE solution}
To establish the existence and uniqueness of the FBSDE solution, we need to make some assumptions on the market parameters. Since the FBSDE is fully-coupled, we will first show that the unique solution exists on a small time interval, then apply a pasting argument to get a unique solution on an arbitrary time interval $ [0, T] $.

\begin{condition}[Competition parameters]\label{condA}
Let $ L_a $ be the norm of matrix $ A $ \eqref{matrix_A}, 
\begin{equation*}
L_a = \sup \bra{\nrm{Ax}_2: x \in \mathbb{R}^N \text{ with } \nrm{x}_2 = 1 }.
\end{equation*}
Suppose the following holds:
\begin{equation*}
L_a^2 e < 1, 
\end{equation*} 
where $ e $ denotes the natural number. 
\end{condition}

\begin{remark}
We estimate the matrix norm as 
\begin{align}
L_a^2 = \nrm{A}_2^2 \leq \nrm{A}_1 \nrm{A}_\infty = \br{\sum_{i = 1}^{N} \frac{\theta^i}{N - \theta^i} } \br{\max_{i} \frac{\theta^i (N - 1)}{N - \theta^i} } \leq& \br{\sum_{i = 1}^{N} \frac{\theta^i}{N - 1}} \br{\max_{i} \frac{\theta^i (N - 1)}{N - 1}} \\\nonumber
&\leq \max_i \br{\theta^i}^2.
\end{align}
So the condition is satisfied if $ \max_i \theta^i < e^{-0.5} \approx 0.6 $. The moderate level of $ \theta^i $ is a reasonable assumption in view of the weak interaction in the existing literature \cite{horst2005stationary}.
\end{remark}

\begin{theorem}
Assume Condition \ref{condA} and that $ (h_t)_{t\leq T} $ is a bounded process. There exists a unique solution to the FBSDE (\ref{mainbsde}).
\end{theorem}

\begin{proof}
Suppose $ \abs{b_t} = \abs{\frac{\hat{h}_t}{\sigma}} \leq \tilde{b} \in \mathbb{R} $. We first show the solution exists on a small interval, i.e, there exists a $ \delta_b $ s.t for $ \delta_c \leq \delta_b $, the FBSDE has a unique solution on $ S_{\delta_c}^2(\mathbb{R}^N) \times S_{\delta_c}^2(\mathbb{R}^N) \times H_{\delta_c}^2(\mathbb{R}^{N, d}) $.
Denote 
\begin{align}
\label{forward_g}
g(t, z) & = zb + \frac{1}{\alpha} \abs{b_t}^2, \\
\label{forward_sigma}
\sigma(t, z) & = z + \frac{1}{\alpha} b_t, \\
\label{backward_f}
f(t, z) & = zb + \frac{1}{2 \alpha} \abs{b_t}^2. 
\end{align}

Let $ \delta_h > 0 $ to be determined and $ \delta_c \in (0, \delta_h] $. Let $ x \in \mathbb{R}^N $ be fixed. We introduce the following norm 
\begin{equation}
\label{Nbar}
\nrm{(Y, Z)}_{\overline{\mathcal{N}}[0,\delta_c]} = \underset{t \in [0, \delta_c]}{\sup} \bra{\mathbb{E}\abs{Y_t}^2 + \mathbb{E}\int_{t}^{\delta_c}\abs{Z_t}^2 ds }^{1/2}.
\end{equation}

Let $ \overline{\mathcal{N}}[0,\delta_c] $ be the completion of $ \mathcal{N}[0, \delta_c] $ in $ H^2_{\delta_c}(\mathbb{R}^N) \times H^2_{\delta_c}(\mathbb{R}^{N, d}) $ under norm (\ref{Nbar}). Take any $ (Y^{(i)}, Z^{(i)}) \in  \overline{\mathcal{N}}[0, \delta_c] $, $ i = 1, 2 $, the SDE for $ X^{(i)} $ is:
\begin{align}\label{smallintforward}
dX^{(i)} &= g(t, Z^{(i)}) dt + \sigma(t, Z^{(i)}) d\zeta_t, \quad t \in [0, \delta_c],\\
X^{(i)}_0 & = x^{(i)}.
\end{align}

Since both $ g $ and $ \sigma $ are independent of $ x $, the above SDE has a unique strong solution $ X^{(i)} \in H^2_{\delta_c}(\mathbb{R}^N)$. Apply It\^{o}'s formula to $ \abs{X^{(1)}_t - X^{(2)}_t}^2 $, we then obtain the following estimate:
\begin{equation*}
\begin{split}
\mathbb{E} \abs{X^{(1)}_t - X^{(2)}_t}^2 & = \mathbb{E} \int_{0}^{t} \br{2 b_s \abs{X^{(1)}_s - X^{(2)}_s} \abs{Z^{(1)}_s - Z^{(2)}_s} + \abs{Z^{(1)}_s - Z^{(2)}_s}^2} ds\\
& \leq \mathbb{E} \pr{\int_{0}^{t} \epsilon_1 \abs{Z^{(1)}_s - Z^{(2)}_s}^2 ds  + \int_{0}^{t} \frac{{\tilde{b}}^2}{\epsilon_1} \abs{X^{(1)}_s - X^{(2)}_s}^2 ds + \int_{0}^{t} \abs{Z^{(1)}_s - Z^{(2)}_s}^2 ds} 
\end{split}
\end{equation*}
for some constant $ \epsilon_1 > 1 $.

By Gronwall's inequality, we have
\begin{equation}\label{diffx_by_diffz}
\mathbb{E} \abs{X^{(1)}_t - X^{(2)}_t}^2 \leq e^{\frac{1}{\epsilon_1}\int_{0}^{\delta_c} \tilde{b}^2 ds} \mathbb{E}\int_{0}^{\delta_c} (1 + \epsilon_1) \abs{Z^{(1)}_s - Z^{(2)}_s}^2 ds.
\end{equation}

Next, we solve the following BSDEs $ (i = 1, 2) $:
\begin{align}\label{smallintbsde}
d\bar{Y}^{(i)} &= f(t, \bar{Z}^{(i)}) dt + \bar{Z}^{(i)} d\zeta_t, \quad t \in [0, \delta_c], \\
\bar{Y}^{(i)}_{\delta_c} &= A X_{\delta_c}^{(i)}.
\end{align}
where $ X_{\delta_c}^{(i)} $ is simulated with process $ Y^{(i)} $ and $ Z^{(i)} $. Recall that we use $ A $ without the time index to denote a constant matrix.

Applying results to BSDEs with random coefficients, we see that the BSDE \eqref{smallintbsde} has a unique adapted solution $ (\bar{Y}^{(i)}, \bar{Z}^{(i)}) \in \mathcal{N}[0, \delta_c] \subset \bar{\mathcal{N}}[0, \delta_c] $. Define a map $ \mathcal{T}: \overline{\mathcal{N}}[0, \delta_c] \to \overline{\mathcal{N}}[0, \delta_c] $ by $ (Y^{(i)}, Z^{(i)}) \mapsto (\bar{Y}^{(i)}, \bar{Z}^{(i)}) $. Apply Ito's formula to $ \abs{\bar{Y}^{(1)}_t - \bar{Y}^{(2)}_t}^2 $, we have 
\begin{equation*}
\begin{split}
& \mathbb{E} \abs{\bar{Y}^{(1)}_t - \bar{Y}^{(2)}_t}^2 + \int_{t}^{\delta_c}  \abs{\bar{Z}^{(1)}_s - \bar{Z}^{(2)}_s}^2 ds \\
& \leq L_a^2 \mathbb{E} \abs{X^{(1)}_{\delta_c} - X^{(2)}_{\delta_c}}^2 + \int_{t}^{\delta_c} 2 b_s \abs{\bar{Y}^{(1)}_s - \bar{Y}^{(2)}_s} \abs{{Z}^{(1)}_s - {Z}^{(2)}_s} ds \\
& \leq L_a^2  e^{\frac{1}{\epsilon_1}\int_{0}^{\delta_c} \abs{b_s}^2 ds} \mathbb{E}\int_{0}^{\delta_c} (1 + \epsilon_1) \abs{Z^{(1)}_s - Z^{(2)}_s}^2 ds + \int_{t}^{\delta_c} 2 b_s \abs{\bar{Y}^{(1)}_s - \bar{Y}^{(2)}_s} \abs{{Z}^{(1)}_s - {Z}^{(2)}_s} ds \\
& \leq L_a^2  e^{\frac{1}{\epsilon_1}\delta_c \tilde{b}^2} \mathbb{E}\int_{0}^{\delta_c} (1 + \epsilon_1) \abs{Z^{(1)}_s - Z^{(2)}_s}^2 ds + \int_{t}^{\delta_c} \frac{\tilde{b}^2}{\epsilon_2} \abs{\bar{Y}^{(1)}_s - \bar{Y}^{(2)}_s}^2 ds + \epsilon_2  \int_{t}^{\delta_c} \abs{{Z}^{(1)}_s - {Z}^{(2)}_s}^2 ds 
\end{split}
\end{equation*}
where we used \eqref{diffx_by_diffz} in the second last inequality. Hence, 
\begin{equation*}
\begin{split}
& \mathbb{E} \abs{\bar{Y}^{(1)}_t - \bar{Y}^{(2)}_t}^2 + \int_{t}^{\delta_c}  \abs{\bar{Z}^{(1)}_s - \bar{Z}^{(2)}_s}^2 ds \\
& \leq e^{\frac{1}{\epsilon_2}\delta_c \tilde{b}^2 } \bra{\pr{\epsilon_2 + L_a^2  e^{\frac{1}{\epsilon_1} \delta_c \tilde{b}^2 } (1 + \epsilon_1) } \mathbb{E}\int_{0}^{\delta_c}  \abs{Z^{(1)}_s - Z^{(2)}_s}^2 ds} + \frac{ \tilde{b}^2}{\epsilon_2} \int_{t}^{\delta_c} \abs{\bar{Y}^{(1)}_s - \bar{Y}^{(2)}_s}^2 ds \\
& \leq e^{\frac{1}{\epsilon_2} \delta_c \tilde{b}^2 } \bra{\pr{\epsilon_2 + L_a^2  e^{\frac{1}{\epsilon_1} \delta_c \tilde{b}^2 } (1 + \epsilon_1) } \mathbb{E}\int_{0}^{\delta_c}  \abs{Z^{(1)}_s - Z^{(2)}_s}^2 ds} + \frac{ \tilde{b}^2}{\epsilon_2} \delta_c \sup_{t \leq \delta_c} \abs{\bar{Y}^{(1)}_t - \bar{Y}^{(2)}_t} \\ 
& \leq C(\epsilon_1, \epsilon_2, \delta_c) \nrm{(Y^{(1)}, Z^{(1)}) - (Y^{(2)}, Z^{(2)})}^2_{\overline{\mathcal{N}}[0, \delta_c]} 
\end{split}
\end{equation*}
for $ C(\epsilon_1, \epsilon_2, \delta_c) = \max \bra{e^{\frac{1}{\epsilon_2} \delta_c \tilde{b}^2} \pr{\epsilon_2 + L_a^2 e^{\frac{1}{\epsilon_1} \delta_c \tilde{b}^2} (1 + \epsilon_1)}, \frac{\tilde{b}^2}{\epsilon_2}} $.

Denote $ C(\epsilon_i) =  e^{\frac{1}{\epsilon_i} \delta_c \tilde{b}^2} $. By choosing
\begin{align*}
\epsilon_1 = \epsilon_2 = 2 \delta_c \tilde{b}^2,
\end{align*}
we have
\begin{align*}
C(\epsilon_1) = C(\epsilon_2) = \frac{1}{2}.
\end{align*}

Therefore,
\begin{equation*}
\begin{split}
C(\epsilon_1, \epsilon_2, \delta_c) & = \max \bra{e^{\frac{1}{2}} \br{\epsilon_2 + L_a^2 e^{\frac{1}{2}}}(1 + \epsilon_1), \frac{1}{2} }  \\
& \leq \max \bra{L_a^2 e + e^{\frac{1}{2}} \epsilon_2 + L_a^2 e  \epsilon_1 + e^{\frac{1}{2}} \epsilon_1 \epsilon_2, \frac{1}{2}} < 1 . 
\end{split}
\end{equation*}
for $ \delta_c $ small enough by condition (\ref{condA}).

By the contraction mapping theorem, there exists a unique fixed point $ (Y, Z) $ for $ \mathcal{T} $. In fact, $ (Y, Z) \in \mathcal{N}[0, \delta_c] $. Let $ X $ be the corresponding solution to the SDE \eqref{smallintforward}, $ (X, Y, Z) \in S_{\delta_c}^2(\mathbb{R}^N) \times S_{\delta_c}^2(\mathbb{R}^N) \times H_{\delta_c}^2(\mathbb{R}^{N, d}) $ is a unique adapted solution of (\ref{smallintbsde}) with forward process (\ref{smallintforward}). 

We next show that a unique solution exists for the problem with an arbitrary time horizon $ T > 0 $ by a pasting method. The argument requires the construction of a decoupling random field, which will be identified with the $ Y $ component in the FBSDE solution.

\vspace{0.5em}
To this end, denote $ \Theta \coloneqq (\bm{X}, \bm{Y}, \bm{Z}) $, and consider a FBSDE on a subinterval $ [t_1, t_2] $:
\begin{align}
\label{subfbsdex}
\bm{X}_t &= \bar{\eta} + \int_{t_1}^{t} g(s, \Theta_s) ds + \int_{t}^{t_1}\sigma(s, \Theta_s) d\zeta_s, \\
\label{subfbsdey}
\bm{Y}_t &= \phi(\bm{X}_{t_2}) - \int_{t}^{t_2} f(s, \Theta_s) ds - \int_{t}^{t_2} Z_s d\zeta_s, \quad t \in [t_1, t_2].
\end{align}
where $ \bar{\eta} \in L^2(\mathcal{F}_{t_1}) $ and $ \phi(x, \cdot) \in  L^2(\mathcal{F}_{t_2})$, for each fixed $ x $. In the Markovian case and if $ \sigma(t, \cdot) $ is independent of $ Z $, the solution $ Y_t = u(t, X_t) $ is a (viscosity) solution to a quasilinear PDE. For non-Markovian case as in the present setting, the solution $ Y_t $ is a random function of the stochastic process. This correspondence provides intuition to the following definition that is standard in the literature (\cite{ma2015well}).

\begin{definition}
A decoupling field of FBSDE (\ref{mainbsde}) is an $ \mathbb{F} $-progressively measurable random field $ u: [0, T] \times \mathbb{R} \times \Omega \mapsto \mathbb{R} $ with $ u(T, x) = h(x) $ if there exists a constant $ \delta_h > 0 $ such that, for any $ 0 = t_1 < t_2 \leq T $ with $ t_2 - t_1 \leq \delta_c $ and any $ \bar{\eta} \in L^2(\mathcal{F}_{t_1}) $, the FBSDE (\eqref{subfbsdex} - \eqref{subfbsdey}) with initial value $ \bar{\eta} $ and terminal condition $ u(t_2, \cdot) $ has a unique solution that satisfies $ Y_t = u(t, X_t) $ for $ t \in [0, T] $. Such decoupling field $ u $ is called \textit{regular} if it is uniformly Lipschitz in the spacial variable $ x $. 
\end{definition}
	
To construct a regular decoupling field, we look at a variational FBSDE. Omitting the subscript $ t $, fix processes $ Y^{(1)}, Y^{(2)}, X^{(1)}, X^{(2)} \in [0, T] \times \mathbb{R}^N $. The initial conditions for $ X^{(1)}, X^{(2)} $ are $ x^{(1)} $ and $ x^{(2)} $, respectively. Let $ i $ and $ j $ be the indices for vector components. 

Define the operator $ \bar{\nabla}: ([0, T] \times \mathbb{R}^N)^2 \to [0, T] \times \mathbb{R}^{N \times N} $ as
\begin{equation*}
(\bar{\nabla} Y)_{i, j} := 
\begin{cases}
&\frac{Y^{(1)}_i - Y^{(2)}_i}{x^{(1)}_j - x^{(2)}_j} \quad \text{ if } x^{(1)}_j \neq x^{(2)}_j, \\
&0 \quad \quad \text{ if } x^{(1)}_j = x^{(2)}_j.
\end{cases}
\end{equation*} 
for $ i, j \in \{1, ..., N \} $. The operator $ \bar{\nabla} $ on a constant vector is defined similarly.

We next show the FBSDE in $ \bar{\nabla} X $, $ \bar{\nabla} Y $ and $ \bar{\nabla} Z $ has a unique solution that is in fact a constant. This solution leads to a decoupling field associated to the original FBSDE. Then by a pasting argument, the FBSDE has a unique solution on an arbitrary time interval $ [0, T] $. 

By equation \eqref{forward_g} and \eqref{forward_sigma}, and let $ z_1 $, $ z_2 $ denote arbitrary vectors in $ \Rbb^N $, we have
\begin{align}
\frac{g(t, z_1) - g(t, z_2) }{z_1 - z_2} & = b_t, \\
\frac{\sigma(t, z_1) - \sigma(t, z_2) }{z_1 - z_2} & = 1 .
\end{align} 

We can compute 
\begin{align}\label{diffbsdex}
X^{(1)} - X^{(2)} & = x^{(1)} - x^{(2)} + \int_{0}^{t} \frac{g^{(1)} - g^{(2)}}{Z^{(1)} - Z^{(2)}} \br{Z^{(1)} - Z^{(2)}}ds + \int_{0}^{t} \frac{\sigma_1 - \sigma^2}{Z^{(1)} - Z^{(2)}} (Z^{(1)} - Z^{(2)}) d\zeta_s, \\ 
\label{diffbsdey}
Y^{(1)} - Y^{(2)} & = A(X_T^{(1)} - X_T^{(2)}) - \int_{t}^{T} \frac{f^{(1)} - f^{(2)}}{Z^{(1)} - Z^{(2)}} (Z^{(1)} - Z^{(2)}) ds - \int_{t}^{T} (Z^{(1)} - Z^{(2)}) d\zeta_s.
\end{align}

It can be verified that \eqref{diffbsdex} - \eqref{diffbsdey} imply that the processes $ \bar{\nabla} X $, $ \bar{\nabla} Y $ and $ \bar{\nabla} Z $ satisfy the following FBSDE  
\begin{align}
\bar{\nabla} X & = \bar{\nabla} x + \int_{0}^{t} b_s \bar{\nabla} Z_s ds + \int_{0}^{t} \bar{\nabla} Z_s d\zeta_s, \\
\bar{\nabla} Y & = A \bar{\nabla} X_T - \int_{t}^{T} b_s \bar{\nabla} Z_s ds - \int_{t}^{T} \bar{\nabla} Z_s d\zeta_s. 
\end{align}
	
A solution to the above FBSDE is $ (\bar{\nabla} X, \bar{\nabla} Y, \bar{\nabla} Z)_{t \in [0, T]} = (\bar{\nabla} x , A \bar{\nabla} x, 0)_{t \in [0, T]} $.

Similarly, define
\begin{equation}
(\bar{\nabla} u(t))_{i, j} = \frac{u^i(t, X_t^{(1)}) - u^i(t, X_t^{(2)})}{X_t^{(1), j} - X_t^{(2), j}}, 
\end{equation}
we must have
\begin{equation*}
\bar{\nabla} u(t) = \bar{\nabla} Y_t (\bar{\nabla} X_t)^{-1} = A.
\end{equation*}
	
The random field $ u(t, x) $ is uniformly Lipschitz continuous in the spacial variable. Hence it is regular.

Let $ \delta_c > 0 $ be small enough so that the FBSDE (\ref{mainbsde}) admits a unique solution $ \Theta \in {L}^2 $ for $ t \leq \delta_c $. For any $ (t, x) $, denote the (unique) solution to FBSDE (\ref{mainbsde}) starting from $ (t, x) $ by $ \Theta^{t, x} $, and denote a random field by $ u(t, x) = Y_t^{t, x} $. The uniqueness of solution to FBSDE then leads to that $ Y_s^{t, x} = u(s, X_s^{t, x}) $, for $ s \in [t, T] $, $ \mathbb{P} $-a.s.

Let $ 0 = t_0 < ... < t_n = T $ be a partition of $ [0, T] $ such that $ t_i - t_{i - 1}  \leq \delta_c $, $ i - 1, ..., n $. We first consider the FBSDE (\ref{mainbsde}) on $ [t_{n-1}, t_n] $. By existence of solution on a small interval, there exists a process $ Y_t^{t_{n-1}, x} $, for $ x = X_{t_{n-1}} $ and hence a random field $ u(t, x) $ for $ t \in [t_{n-1}, t_n] $ such that $ \bar{\nabla} u(t) = A $ for all $ t \in [t_{n-1}, t_n] $. Next, consider FBSDE (\ref{mainbsde}) on $ [t_{n-2}, t_{n-1}] $ with terminal condition $ u(t_{n-1}, \cdot) $. Apply the results on solution on small interval, we find $ u $ on $ [t_{n-2}, t_{n-1}] $ such that $ \bar{\nabla} u(t) = A $ for $ t \in [t_{n-2}, t_{n-1}] $. Repeating this procedure backward $ n $ times, we extend the random field $ u $ to the whole interval $ [0, T] $.

We now show the solution obtained in this way is in the right space. 
	
Define 
\begin{equation*}
\begin{split}
I_0^2 & \coloneqq \mathbb{E}\bra{\br{ \int_{0}^{T} \br{|g| + |f|}(s, 0) ds}^2 + \int_{0}^{T} |\sigma(s, 0)|^2 ds } \\
& \leq \br{\mathbb{E}  \int_{0}^{T} |b_s|^2 ds}^2 + \mathbb{E}  \int_{0}^{T} |b_s|^2 ds \\
& \leq T^2 \tilde{b}^4 + T \tilde{b}^2 < \infty. 
\end{split}
\end{equation*}

By $ |u(t, x)| \leq |u(t, 0)| + |x| $, considering the FBSDE on each interval $ [t_i, t_{i+1}] $ with initial value $ X_{t_i} = 0 $, we see that there exists a constant $ C $ such that 
\begin{equation}\label{u_bound}
\mathbb{E} |u(t_i, 0)|^2 = \mathbb{E}\abs{Y_{t_i}^{t_i, 0}}^2 \leq C \br{\mathbb{E} |u(t_{i+1}, 0)|^2 + \mathbb{E} \abs{X_{t_{i+1}}}^2} + CI_0^2.
\end{equation}

Since $ u(t_n, 0) = 0 $, we have 
\begin{equation}
\underset{0\leq i \leq n}{\max} \mathbb{E} |u(t_i, 0)|^2 \leq CI_0^2.
\end{equation}
	
A standard estimation using the forward and backward dynamics and the bounds for the coefficients yields,
\begin{align}\label{x_bound}
\begin{split}
& \mathbb{E}\bra{\underset{t_i \leq t \leq t_{i+1}}{\sup} \br{|X_t|^2 + |Y_t|^2} + \int_{t_i}^{t_{i+1}}|Z_s|^2 ds} \\
& \leq C \mathbb{E}\pr{\abs{X_{t_i}}^2 + \abs{u(t_{i+1}, 0)}^2} + CI_0^2
\end{split}
\end{align}
	
To estimate $ |X_{t_i}|^2 $, notice that \eqref{x_bound} and \eqref{u_bound} imply that 
\begin{equation*}
\mathbb{E} |X_{t_{i+1}}|^2 \leq C \mathbb{E}\pr{|X_{t_i}|^2 + |u(t_{i+1}, 0)|^2} + CI_0^2 \leq C\mathbb{E} |X_{t_i}|^2 + CI_0^2.
\end{equation*}
	
Therefore, 
\begin{equation}
\begin{split}
& \mathbb{E}\bra{\underset{0 \leq t \leq T}{\sup} \br{\abs{X_t}^2 + \abs{Y_t}^2} + \int_{0}^{T}\abs{Z_s}^2 ds} \\
& \leq C \br{\abs{x}^2 + I_0^2}.
\end{split}
\end{equation}
\end{proof}

\begin{remark}
The proof relies crucially on the boundedness of $ (b_t)_{t \in [0, T]} $, or equivalently, of the return rate $ (h_t)_{t \in [0, T]} $. The case where $ h_t $ is an unbounded stochastic process is still open and it is left for future research.
\end{remark}

\section{Deep learning algorithms}
In this section, we will give a brief introduction to the deep neural network that will be used in our numerical scheme. 

\subsection{The neural network as function approximators}
Neural networks are compositions of simple functions. They are efficient in approximating the solutions of (stochastic) differential equations. To obtain a good approximator, it usually requires the algorithm to find the best parameters in the function composition, which, in many cases, is convenient by the method of SGD.

We adopt notations from \cite{hure2019some} and consider simple feedforward neural networks (NNs). Denote the dimension of state variable $ x $ by $ d^x $. Fix a input dimension $ d^I = d^x $ if the approximated function is only in the variable $ x $. We may take time $ t $ as an additional input parameter to enable parameter sharing across time steps. In this case, the solution at all time steps is modeled by a single neural network and the function will depend on $ (t, x) $ and $ d^I = d^x + 1 $. We denote the output dimension by $ d^O $, and $ d^O = N $ for $ N $ the dimension of the FBSDE solution. There are total number of $ L + 1 \in \mathbb{N} $, $ L \geq 2 $ layers for each NN, with $ m_l, l \in \{0, ..., L\} $, the number of neurons in each hidden layer. For simplicity, we choose $ m_l = m $ for $ l \in  \{1, ..., L-1\} $.
\par 
More specifically, the fully connected feedforward neural network for the FBSDE solver is a function from $ \mathbb{R}^{d^I} $ to $ \mathbb{R}^{d^O} $  defined by the composition map
\begin{equation*}
x \mapsto P_L \cdot \varphi \cdot  P_{L-1} \cdot ... \cdot \varphi \cdot P_1(x) \coloneqq f(x) \in \mathbb{R}^N.
\end{equation*}
for $ x \in \mathbb{R}^{d^I} $. Here, $ P_l $, $ l \in \{1, ..., L\} $ are affine functions with assigned input and output dimensions. To be specifical on the structure of $ A_l $, 
\begin{equation*}
P_l(x) = w_l x + b_l, \quad l \in \{1, ..., L\}.
\end{equation*} 

The matrix $ w_l $ and vector $ b_l $ is the weight and bias of a hidden layer, respectively. $ \varphi: \mathbb{R} \to \mathbb{R} $ is the activation function, which is a nonlinear function that can be customized. Some standard activation functions are $ \tanh $, Softmax, Sigmoid, and ReLU. For all our experiments, we use $ \text{ReLU}(x) = \max\{0, x\} $ as the activation function for the fully connected networks.

Denote the parameters of the neural network by $ \theta \in \mathbb{R}^{N_\theta} $, which includes all the matrices $ w_l $ and vectors $ b_l $. Let $ N_\theta(m) = \sum_{l=0}^{L-1}m_l(1 + m_{l+1}) = d^I(1 + m) + m(1 + m) (L-1) + m (1 + N) $. 
Denote $ \Theta_m $ the set of possible parameters with $ m $ hidden units. 

The neural network that satisfies the given input and output dimension, the number of layers, and with the nonlinear function $ \varphi $ is in the function space 
\begin{equation*}
\mathcal{N}\mathcal{N}^\varphi_{d^I, d^O, L} = \underset{m \in \mathbb{N}}{\cup} \mathcal{N}\mathcal{N}^\varphi_{d^I, N, L, m} (\Theta_m) = \underset{m \in \mathbb{N}}{\cup} \mathcal{N}\mathcal{N}^\varphi_{d^I, N, L, m} (\mathbb{R}^{N_\theta(m)})
\end{equation*}

By a learnable variable, we mean any variable that is needed to compute the value of the loss function and can be optimized, other than the parameters in the above functional form of neural networks.

\subsubsection{The recurrent neural network}
We will use the recurrent neural network for the network-based estimation step, and we introduce it here. Denote the weights and bias parameters similar as before. Let $ \varphi $ denote the activation function. the details network structure is as follows. For a time series sequence $ x = (x_0, ..., x_t, ..., x_T) $, we compute the hidden state at time $ t $, $ H_t $ inductively by 
\begin{align}
H_t &= \varphi (w_{it} x_t + b_{it} + w_{ht} H_{t-1} + b_{ht}), \\
f_t &= \varphi(w_t H_t + b_t), \quad t \in \{0, ..., T\}.
\end{align}
where the subscript $ it $ denotes the weight and bias for input at time $ t $ and the subscript $ ht $ indicates the weight and bias for the hidden state at time $ t < T $. The parameters $ w_t $ and $ b_t $ denotes the weight and bias of the linear map for the time $ t $ hidden state. The RNN structure we use for Stage I estimation is exactly this one, with $ x_t $ being the stock price at time $ t $.

\subsection{Deep Learning Scheme}
We perform the deep learning scheme on several independent models. Each model corresponds to a particular market setting. 
\textbf{HM} indicates the homogeneous initial belief, \textbf{HT} for the heterogeneous initial belief, \textbf{PI} for partial information, \textbf{FI} for full information. \textbf{L} for the case of linear Gaussian filter. \textbf{NL} for nonlinear filter. \textbf{C} for game with competition, and \textbf{NC} for no competition.

We use the uniform time discretization for interval $ [0, T] $. Let $ 0 = t_0 < ... < t_K = T $ be such that $ \Delta t = t_k - t_{k - 1} $, $ k \in \{1, ..., K \} $. The conditional expectation in the previous section is a function of stock prices and the initial belief that best approximates the conditioned variable in the least square sense. The deep learning scheme is performed in two stages. In Stage I, we approximate the conditional expectation of the mean stock return as a function of the hidden variable, $ h_t = h(A_t) $ on $ \mathcal{F}_t^S $. The hidden state $ H_t $ computed by the feed-forward neural network depends the past stock prices up to time $ t $, as well as the initial prior of the investor's estimate on the market return. Therefore, so does the approximated estimation of the investor, $ \hat{h} $ at time $ t $. We optimize $ E = 3 $ independent networks in parallel and take the average of network outputs to get a single agent's estimation. The variance reduction technique of averaging random outcomes is common in the classical Monte Carlo method. Let $ \mathcal{G}_k $ be the neural network approximation of conditional expectation at time step $ k $. The input variable is the all the asset prices at discrete time steps, including that at time $ 0 $, and the investor's initial prior. Hence, $ d^{I, K} = d (K + 1) + 1 $ where $ d $ is the number of stocks and $ K $ is the maximal time step index. Let $ \mathcal{G} \in {Rnn}^\varphi_{d^{I,K}, N, m=64} $, and $ G_k $ be the $ k $-th output in the sequential order that depends on information up to time $ t_k $. For mini-batch of size $ B $, the loss function for Stage I is 
\begin{equation*}
LossI =  \frac{1}{B} \frac{1}{K+1} \sum_{j = 1}^B \sum_{k = 0}^K \nrm{\mathcal{G}_k(S^{(j)}_{\cdot \wedge t_k }, \hat{A}_0^{(j)}) - h_k^{(j)}(\hat{A}_0^{(j)}) }^2
\end{equation*}
where $ S^{(j)}_{\cdot \wedge t_k } $ denotes the stock prices up to time $ t_k $, and $ h_k^{(j)}(\hat{A}_0^{(j)}) $ indicates the $ j $-th simulated stock return from the investor's subjective probability measure $ \mathbb{P}^i $ mainly caused by different initial beliefs. Suppose $ \mathcal{G}^{*, (e)} $ is the trained model for the $ e $-th independent network, the estimation of the stock return at time $ t_k $ given the stock price path $ S $ and initial estimate $ \hat{A}_0 $ is 
\begin{equation*}
\hat{h}_k = \frac{1}{E} \sum_{e=1}^{E} \mathcal{G}_k^{*, (e)}(S_{\cdot \wedge t_k}, \hat{A}_0).
\end{equation*}

In Stage II, we use neural networks to approximate the solution $ (X_t, Y_t, Z_t), \; t \in [0, T] $ of the FBSDE. Let $ \hat{Y}_0 $ be learnable variables that will be optimized by the SGD. $ \hat{X}_0 = x $. First, define the vector $ \bm{h} $ by $ \br{\bm{h}_k}_i = h_k^{i} $ where $ h_k^{i} $ is the agent $ i $'s estimation of return at time $ t_k $. The network input is $ (\hat{\bm{h}}_k, S_k, t_k) $, which is a vector consisting the value of estimated returns, the stock prices and a time variable. Let $ d^I = 2 d N + 1 \coloneqq d^S $ (S for Solver) for $ d $ the number of stocks as before, $ \mathcal{G}^S \in \mathcal{N} \mathcal{N}_{d^S, N, 3} $, $ \hat{h}_k, S_k \in \mathbb{R}^d $. Compute $ \hat{Y}_{k+1} $ from $ \hat{Y}_k $ by (omitting index $ j $ for $ j \in \{1, ..., B \} $)
\[ \hat{Y}_{k+1} = \hat{Y}_k + \hat{Z}_k \Delta \zeta_k + f(\hat{\bm{h}}_k, \hat{Z}_k) \Delta t, \]
and $ \hat{X}_{k+1} $ from $ \hat{X}_k $ by 
\[ \hat{X}_{k+1} = \hat{X}_k + \br{\hat{Z}_k \frac{\hat{\bm{h}}_k}{\sigma} + \frac{1}{\bm{\alpha}} \abs{\frac{\hat{\bm{h}}_k}{\sigma}}^2} \Delta t + \br{\hat{Z}_k + \frac{1}{\bm{\alpha}} \frac{\hat{\bm{h}}_k}{\sigma} } \Delta \zeta_k  .\]
where $ \hat{Z}_k = \mathcal{G}^S(\hat{\bm{h}}_k, S_k, t_k) $, until $ k = K $. The loss function is 
\begin{equation*}
LossII = \frac{1}{B} \sum_{j=1}^B \nrm{\hat{Y}_K^{(j)} - A \hat{X}_K^{(j)}}^2.
\end{equation*}

We use a single layer recurrent neural network (RNN) with hidden units $ m = 64 $ for Stage I, the estimation step. The choice of network structure utilizes the time series nature of the input variable, and at the same time to reduce computational complexity. Due to the path dependence of the estimation, the estimate $ \hat{h}_k $ at different times require neural network approximators of varying input dimension, if without the RNN. The RNN takes the sequence of stock prices indexed by time as input and outputs a sequence of hidden states indexed by time. Each element of the output sequence depends on the data up to the index time. We then transform each hidden state by a linear map to obtain the estimation of return at the corresponding time. The loss is the mean square error (MSE) of the estimate against the true return. To minimize the effort of hyperparameter tuning that gives no structural changes, we use the default activation function from the pytorch RNN module and let $ \varphi = \tanh $.

For Stage II, we use networks with $ L = 3 $ layers and $ m = 64 $ hidden units for $ \mathcal{G}^S $. We use the Adam optimizer in both stages. The initial rate for Stage I is $ lr = 1e^{-3} $ and we use learning rate decay and half the learning rate every $ ep^{decay} = 400 $ steps. The learning rate for Stage II is $ lr = 3e^{-3} $, which is a standard choice for solving the BSDEs. Learning rates on the same scale produce similar results. We include a comparison of the same deep learning scheme with different learning rates in the appendix. As mentioned in the previous section, we use ReLU as the activation function for the fully connected network in the FBSDE solver.

The training proceeds with $ ep^{train} = 5000 $ epochs for Stage I, followed by $ ep^{train} = 5000 $ epochs for Stage II. Mini-batch size is $ B = 64 $ for both stages. The deep learning scheme is efficient and robust across different sets of hyperparameters. We use deeper networks on Stage I compared to that of Stage II because of the path dependence nature of the estimation problem. Since the network is long, fewer hidden units in each layer are needed to achieve the same complexity of the function approximator. In all the numerical experiments, we fix the investment horizon $ T = 0.5 $ year. The CPU time for Stage II (FBSDE solver) of the parameter set $ (\eta, \bar{\mu}) = (0.02, 0) $, with $ 5000 $ training epochs is $ 468s \approx 8 \text{mins} $ on a MacBook Pro with the 2.2 GHz Quad-Core Intel Core i7 processor.

\section{Numerical results and model implications}\label{numerical_results}
Although the solution to the PDE is in an analytic form, we still need to compute the values of integrals by numerical integrations. In this section, we present the solution to the HJB equation, as well as the numerical solution by solving the multi-dimensional FBSDE (\ref{mainbsde}) using the deep learning method. We compare results from both methods in case of linear filters when PDE solutions are available. We further apply the deep learning scheme on FBSDEs when nonlinear filters are used to obtain the estimate of return rate. When the estimated return $ \hat{h}_t $ is bounded, the solution to the FBSDE is unique. The deep learning solution converges to the unique Nash equilibrium. When the return process is not necessarily bounded, which in our case, can be a CIR process for the hidden variable $ A_t $ and a square root relation between the mean return and the hidden variable, we do not have theoretical results on the uniqueness of the solution. However, we can still apply the numerical method to find an equilibrium strategy. 

\subsection{Linear filter: homogeneous initial belief}
In this section, we assume the investor's initial estimate is accurate, i.e, $ \hat{\Sigma}(0) = 0$. 
Denote $ \Delta W_{t_k} = W_{t_{k+1}} - W_{t_k} $ and $ \Delta B_{t_k} = B_{t_{k+1}} - B_{t_k} $. By dynamics \eqref{ou_hidden} - \eqref{ou_stock}, generate sample paths of stock prices by
\begin{align}
\label{linear_hidden}
{A}_{t_{k+1}} &= A_{t_k} - \lambda (A_{t_k} - \bar{\mu}) \Delta t + \sigma_\mu \Delta B_{t_k}, \\
\label{linear_stock}
\frac{S_{t_{k+1}} - S_{t_k}}{S_{t_k}} &= A_{t_k} \Delta t + \sigma_S \br{\sqrt{1 - \rho^2} \Delta W_{t_k} + \rho \Delta B_{t_k}}, \quad k \in \{0, ..., K - 1\}.
\end{align}

Set the base market parameters for the case of linear filters to be 
\begin{equation}\label{lin_params}
\lambda = 8, \sigma_S = 0.15, \sigma_\mu = 0.3, \rho = -0.8,
\end{equation}
where we allow the initial condition for $ h_0 = h(A_0) $ vary, as well as the $ \bar{\mu} $ vary for different experiments. We specify the cases later when presenting the correponding investment strategies.

The risk preference parameters are shown in Table \ref{tab:risk_params}.
\begin{table}[H]
\centering
\begin{tabular}{ccccccc}
	\toprule
	Case & $ \delta_1 $ & $ \delta_2 $ & $ \delta_3 $ &  $ \theta_1 $ &  $ \theta_2 $ & $ \theta_3 $ \\
	\midrule
	\textbf{NC} & 2 & 3 & 5 & 0 & 0 & 0 \\
	\textbf{C} & 2 & 3 & 5 & 0.2 & 0.5 & 0.2 \\
	\bottomrule
\end{tabular}
\caption{Investors' risk parameters with or without competition. Case C indicates investment under wealth competition, and NC indicates the standard CARA utility case.}
\label{tab:risk_params}
\end{table}

To illustrate the investor's estimated return process, we plot the sample path of estimates from the RNN structure in \figref{fig:h_mean}, together with the true return process. The estimate is the average of three independent RNN network approximates. Consistent with standard results from filtering theory, the estimates exhibits a trend that is similar to to the true process but with a time lag. The values of the estimate process are less extreme due to the effect of estimation.

Table \ref{tab:pde_fbsde_solutions_v_initial_pi} is a comparison of the the deep learning results to the benchmark solutions from the PDE method. Both the initial positions and the value functions are accurate for the experimented parameter sets. The largest relative error of the initial position is $ (5.82 - 5.58) / 5.58 = 4.3\% $. All the absolute errors are less than $ 0.2 $ with one exception smaller than $ 0.25 $. The initial values are accurate to the 2nd significant digit, and the maximal relative error in the value is less than $ 0.5\% $. 

\figref{fig:pixlinhomo} exhibits sample path strategies and the wealth processes of the agents. The portfolios were solved for the same three agents as in the previous experiments in all the subfigures. For a clear layout, we only plot the investment strategy $ \pi_t $ of Agent 1 and Agent 3, and omit Agent 2 in the first row. In the second row, the wealth processes for all three agents are shown. In the first row, the black dash-dot line indicates the Merton strategy under the competition case, that is, the strategy by assuming deterministic return process, which strategy can be solved as in \cite{lacker2019mean}. The difference between the investment strategy and the Merton strategy is the hedging demand of the investor under competition utility. The hedging demand vanishes as time approachs the end of the investment horizon, regardless of the market or risk parameters. 

Table \ref{l_mean_std_cv_ratio} shows the statistics of investment strategies for all three investors under different market parameters and risk preferences. Each mean and standard derivation the empirical statistics of $ B = 64 $ sample paths. The point here is not to estimate the true time series mean and std using Monte Carlo method, but to illustrate the distribution of strategies, hence a sample size the same as the training mini-batch size is used. We compute the CV (coefficient of variation) as the ratio of the std and the mean, or the std per unit of the mean, as an additional indicator for the time series volatility of strategies, and we report the average of CVs for the three investors in the table. Investors' strategies are more volatile under full information by observing the std and the CV. However, the CV for the first set of market parameter indicates that the variation per unit of the mean may increase under the partial information setting. Competition does not have a significant effect on the CVs for the three test market parameters, indicating that standard deviation increase mostly due to the increase of the strategy in term of absolute value.

Table \ref{tab:lin_sharpe_vrr} reports the empirical Sharpe ratio and the VRR, which is defined as the mean return over variance, instead of over the std, following \cite{espinosa2015optimal}. We also report the Sharpe ratio and the VRR by viewing the three agents portfolio as the social portfolio. In the current experiments, the mean reverting ($ \bar{\mu} = 0, 0.02 $) level is quite low, since we take a conservative view of the market, which causes both the Sharpe and VRR to be small. For the first set of parameters (the top panel), small variation in the portfolio returns compensates for the inaccurate of return estimates, resulting in the larger Sharpe and VRR in the case PI compared to the case FI , for all three agents. Comparing the portfolio performances for the C and NC case under partial information, whether competition increases or decreases the Sharpe and VRR depends on the interaction of market and investors' parameters. We leave it as future research to find the condition of the parameters and risk preferences that induces each case.



\begin{figure}[tb]
	\centering
	{
	{{\includegraphics[width=0.8\textwidth]{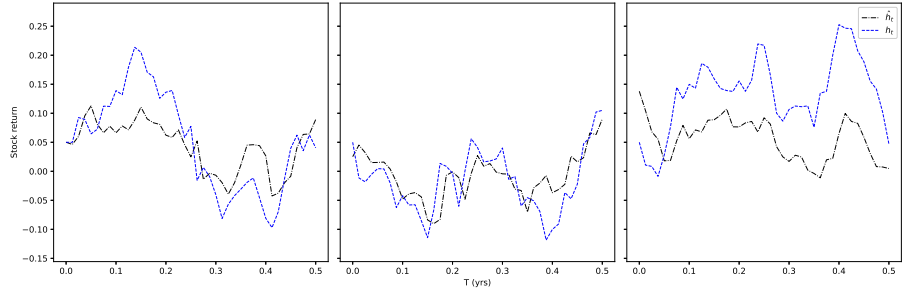}}
	}
	}
	\caption{Sample paths of the estimated return versus the real market return for the linear Gaussian return dynamics with process parameters specified in \eqref{lin_params}. Black dashed line indicates the average of $ 3 $ independent neural network approximations. $ h_0 = 0.05 $, $ \bar{\mu} = 0.02 $, and the blue dashed line indicates the true return. The initial estimate is a constant that equals to the true raturn rate, i.e, $ \hat{h}_0 = 0.05 $. $ x $-axis indicates time.  $ y $-axis is the stock return rate.}
	\label{fig:h_mean}
\end{figure}

\begin{figure}[tb]
	\centering
	{
		{{\includegraphics[width=0.8\textwidth]{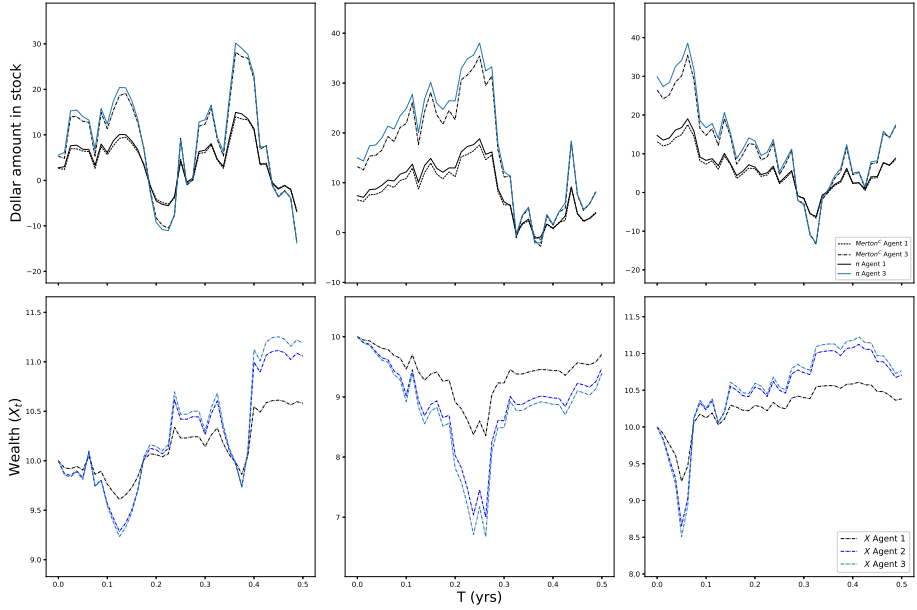}}
		}
	}
	\caption{Sample paths of the dollar amount invested in the stock for Agent $ 1 $ and Agent $ 3 $, as well as the wealth processes $ X_t $ for all three agents. The left column corresponds to investors' response when the market parameters are $ h_0 = 0.02 $, $ \bar{\mu} = 0 $. The middle column corresponds to the investors' response when the market has $ h_0 = 0.05 $, $ \bar{\mu} = 0.02 $. And the right column corresponds to the strategies and wealth processes for investors' when the true market parameters are $ h_0 = 0.1 $ with $ \bar{\mu} = 0.02 $.}
	\label{fig:pixlinhomo}
\end{figure}

\begin{table}[t]
	\resizebox{0.8\textwidth}{!}
	{\begin{tabular}{c c c c c}
			\toprule
			$ \eta $ & $ \bar{\mu} $ & Equation & Initial position & $ V(0) $  \\
			\midrule
			$ 0.02 $ & $ 0 $ & PDE  & $ (2.8736, 5.2312, 5.8239) $  & $ (-0.01773, -0.1829, -0.1955) $ \\
			& & FBSDE & $ (2.7492, 5.1297, 5.5849) $ & $ (-0.01779, -0.1831, -0.1957) $ \\
			\hline 
			$ 0.05 $ & $ 0.02 $ & PDE  &  $ (7.4893, 13.6344, 15.1836) $ & $ (-0.01760, -0.1815, -0.1937) $ \\
			& & FBSDE & $ (7.3917, 13.4794, 14.9981) $ & $ (-0.01759, -0.1814, -0.1939) $ \\
			\hline
			$ 0.1 $ & $ 0.02 $ & PDE  &  $ (14.6734, 26.7123, 29.7434) $ & $ (-0.01740, -0.1796, -0.1912)$ \\
			& & FBSDE & $  (14.7277, 26.9211, 29.8841) $ & $ (-0.01730, -0.1783, -0.1907) $ \\
			\bottomrule
	\end{tabular}}
	\caption{The investors' initial positions and values of investment obtained from solving the PDEs and the FBSDEs. The top, middle and bottom panels shows solution under $ 3 $ different market environment with different initial return and mean-reverting level. The investors' initial prior is a constant equal to the market true return rate.}
	\label{tab:pde_fbsde_solutions_v_initial_pi}
\end{table}

\begin{table}[t]
	\resizebox{0.95\textwidth}{!}
	{\begin{tabular}{@{}llllllllll@{}}
			\toprule
			 & \multicolumn{3}{c}{Mean of (abs. value) strategies} & \multicolumn{3}{c}{Std of (abs. value) strategies} & CV & \multicolumn{2}{l}{Ratio of CV} \\
			 \midrule
			& \multicolumn{1}{c}{\textbf{Agent 1}} & \multicolumn{1}{c}{\textbf{Agent 2}} & \multicolumn{1}{c}{\textbf{Agent 3}} & \multicolumn{1}{c}{\textbf{Agent 1}} & \multicolumn{1}{c}{\textbf{Agent 2}} & \multicolumn{1}{c}{\textbf{Agent 3}} & Mean & \multicolumn{1}{c}{\textbf{PI / FI}} & \multicolumn{1}{c}{\textbf{C / NC}} \\ 
			\midrule
			\textbf{NC-FI} & 5.360 & 8.044 & 13.402 & 3.507 & 5.264 & 8.776 & 0.654 &  &  \\
			\textbf{NC-PI} & 3.562 & 5.345 & 8.908 & 2.422 & 3.630 & 6.056 & 0.680 & 1.038 &  \\
			\textbf{C-PI} & 4.994 & 9.103 & 10.077 & 3.305 & 6.022 & 6.674 & 0.662 &  & 0.974 \\
			\midrule
			\textbf{NC-FI} & 5.763 & 8.647 & 14.397 & 3.726 & 5.590 & 9.311 & 0.647 &  &  \\
			\textbf{NC-PI} & 4.315 & 6.472 & 10.786 & 2.516 & 3.777 & 6.291 & 0.583 & 0.902 &  \\
			\textbf{C-PI} & 6.436 & 11.723 & 12.983 & 3.720 & 6.775 & 7.497 & 0.578 &  & 0.991 \\
			\midrule
			\textbf{NC-FI} & 6.562 & 9.835 & 16.394 & 4.070 & 6.103 & 10.175 & 0.620 &  &  \\
			\textbf{NC-PI} & 5.475 & 8.211 & 13.687 & 3.023 & 4.538 & 7.565 & 0.553 & 0.890 &  \\
			\textbf{C-PI} & 7.707 & 14.037 & 15.542 & 4.446 & 8.096 & 8.966 & 0.577 &  & 1.044 \\ \bottomrule
	\end{tabular}}
	\caption{Time series mean and standard deviation of the three agents's absolute value of investment strategies under different market parameter sets for the linear Gaussian case. The return dynamic is given by (\eqref{linear_stock} - \eqref{linear_hidden}) and parameters are given by \eqref{lin_params}. The CV (coefficient of variation) is the time series standard deviation per unit of the mean, i.e, CV = Std / mean. The ratio PI / FI is the ratio of the CVs of the case PI-NC and the FI-NC. The ratio C / NC is the ratio of CVs of the case PI-C and PI-NC. The top, middle and bottom panel corresponds to the case $ (h_0, \bar{\mu}) $ equals to $ (0.02, 0) $, $ (0.05, 0.02) $ and $ (0.1, 0.02) $, respectively. The investor's initial belief is equal to the true return for all cases.}
	\label{l_mean_std_cv_ratio}
\end{table}

\begin{table}[t]
	\resizebox{0.9\textwidth}{!}
	{\begin{tabular}{clllllllll}
		\toprule
		& \multicolumn{2}{l}{Sharpe ratio} &  &  & VRR &  &  &  &  \\
		\midrule
		& \textbf{Agent 1} & \textbf{Agent 2} & \textbf{Agent 3} & \textbf{Social} & \textbf{Agent 1} & \textbf{Agent 2} & \textbf{Agent 3} & \textbf{Social} \\
		\midrule
		\textbf{NC-FI} & 0.03690 & 0.03695 & 0.03677 & 0.008369 & 0.06130 & 0.04091 & 0.02441 & 0.002779 &  \\
		\textbf{NC-PI} & 0.04776 & 0.04802 & 0.04807 & 0.01357 & 0.1267 & 0.08520 & 0.05124 & 0.007223 &  \\
		\textbf{C-PI} & 0.04645 & 0.04654 & 0.04651 & 0.01182 & 0.08331 & 0.04583 & 0.04138 & 0.004382 &  \\
		\hline 
		\textbf{NC-FI} & 0.06655 & 0.06617 & 0.06613 & 0.01889 & 0.1448 & 0.09614 & 0.05757 & 0.008226 &  \\
		\textbf{NC-PI} & 0.05290 & 0.05285 & 0.05280 & 0.01523 & 0.1497 & 0.09974 & 0.05987 & 0.008628 &  \\
		\textbf{C-PI} & 0.02738 & 0.02743 & 0.02705 & 0.002783 & 0.03355 & 0.01846 & 0.01639 & 0.000704 &  \\
		\hline 
		\textbf{NC-FI} & 0.06192 & 0.06171 & 0.06189 & 0.01639 & 0.10203 & 0.06771 & 0.04078 & 0.005400 &  \\
		\textbf{NC-PI} & 0.02417 & 0.02406 & 0.02418 & 0.003781 & 0.04106 & 0.02723 & 0.01642 & 0.001284 &  \\
		\textbf{C-PI} & 0.03884 & 0.03888 & 0.03904 & 0.007006 & 0.04524 & 0.02485 & 0.02252 & 0.001685 &  \\
		\bottomrule
	\end{tabular}}
\caption{The empirical Sharpe ratio and VRR for each agent as well as the total wealth of agents (social wealth). The top panel (row 1 - 3) corresponds to the market parameters with initial return rate $ 0.02 $ and mean reverting to $ 0 $. The middle panel (row 4 - 6) corresponds to initial return $ 0.05 $ and mean reverting level $ 0.02 $. The bottom panel (row 7 - 9) is for different information setting when the initial market return is $ 0.1 $, with mean-reverting to $ 0.02 $. In each market setting, the investors have the initial belief that is a constant equals to the true market return rate.}
\label{tab:lin_sharpe_vrr}
\end{table}

\subsection{Linear filter: heterogenous prior beliefs}
Recall the estimated return of Agent $ i $ is given by $ \mathbb{E}^i\pr{h(t) | \mathcal{F}^S \vee \hat{A}_0^i } $, where $ \mathbb{E}^i $ indicates expectation under the subjective probability measure $ \mathbb{P}^i $ of Agent $ i $. In both deep learning stages, we need to generate sample paths of $ h_t $ under those probability measures. The initial beliefs in the hidden variable $ \hat{A}_0 $ are sampled from a normal distribution of $ N(m^i, v^i) $, $ i \in \{1, ..., N\} $. Notice that $ v^i = \hat{\Sigma}^i(0) $. 

To focus on the variation in the estimates' accuracy, we assume the means of estimates are the same for all 3 agents and let the standard deviation vary. The mean is equal to the true return rate. More specifically, the parameters for initial estimates are
\begin{equation}\label{hetero_parameter}
(m^1, std^1) = (0.05, 0.05) ,  (m^2, std^2) = (0.05, 0.1), (m^3, std^3) = (0.05, 0),
\end{equation}
where the Agent $ 3 $ has the accurate estimate.

To generate sample paths under the subjective probability measure, we first sample $ \hat{A}^i $ from the prescribed distribution, then proceed as follows:
\begin{align}
\label{A_discrete}
\hat{A}^i_{t_{k+1}} = \hat{A}^i_{t_k} - \lambda (\hat{A}^i_{t_k} - \bar{\mu}) \Delta t + \sigma_a \Delta B_{t_k}, \quad k \in \{0, ..., K - 1\}.
\end{align}

The above equation requires simulation of $ \Delta B $. To obtain the stock prices in $ \mathbb{P} $, we next simulate $ \Delta W $. Let $ A_0 $ be the accurate market return, we simulate according to \eqref{linear_hidden} and \eqref{linear_stock} to get the stock prices $ S_{t_k} $ in the objective world. The estimation is the projection on the subjection view $ \hat{A}^i $.

The hidden state of the recurrent network $ \mathcal{G}^S $ at time $ t_k $ depends on the stock path up to time $ t_k $, as well as $ \hat{A}_0 $. To get the return estimates for investor $ i $, we then optimize the NNs by SGD on the mean square loss of the NN outputs against the subjective hidden state $ \hat{A}^i $. The estimated return is the output of Stage I. The estimation is a part of the neural network input at Stage II for solving the FBSDE. 

To facilitate the comparison, the numerical results for this section in both the case with and without competition, C and NC, respectively, are shown in the later section together with the heterogeneous agents case with nonlinear returns. 



\subsection{Nonlinear filter}
The linear relation between the return rate and stock fundamentals allows us to obtain an explicit solution. However, the assumption is restrictive and unrealistic. To find an investment strategy that is useful in practice, or to derive relevant asset pricing implications from the investment strategies, we need to consider the case of nonlinear $ h $. For the numerical experiments, we focus on the following stock and hidden state dynamics:
\begin{align}\label{nl_stock}
\frac{dS_t}{S_t} &= c \cdot \text{sign}(A_t) \sqrt{|A_t|} dt + \sqrt{1 - \rho^2} \sigma_S dW_t + \rho \sigma_S dB_t \quad (\text{observed}), \\
\label{nl_hidden}
dA_t &= -\lambda (A_t - \bar{\mu})dt + \sigma_a \sqrt{(A_t - a_l)(a_u - A_t)}dB_t \quad (\text{hidden})
\end{align} 
for $ \lambda, c, a_l,  a_u \in \Rbb $.

The numerical scheme for solving the equilibrium investment strategy and value functions is similar to the one with linear filters. The difference is that simulations of $ \hat{A}^i_t $ and $ A_t $ are according to the discretized equation for the above dynamics instead of the Ornstein-Uhlenbeck process \eqref{linear_stock} - \eqref{linear_hidden}.

We next present numerical results in cases of nonlinear filters for heterogeneous market investors with parameters in \eqref{hetero_parameter}. The base market parameters for \eqref{nl_hidden} - \eqref{nl_stock} are
\begin{equation}\label{nl_params}
c = 0.25, \rho = -0.8, \sigma_S = 0.15, \lambda = 1, \sigma_a = 0.4, a_l = -0.3,  a_u = 0.3.
\end{equation}

Similar to the case of linear Gaussian return, \figref{fig:pixnlhomo} shows sample path strategies and the wealth processes of the agents. The black dash-dotted line in the first row of the figure is the Merton strategy that assumes the return is a deterministic process, which we use as the benchmark. And the difference between the real strategies and the benchmark is the hedging demand for stochastic returns. As time approaches the end of the investment horizon, the hedging demand vanishes. 

\figref{fig:picncnl} illustrates the effect of competition. The subfigures are the investment strategies for all the three agents. Each subfigure includes the strategy for an investor in case of with and without relative concerns, C and NC. The horizontal line indicates the time series mean of strategies in the C and NC case. In case with competition, the weight parameters are 0.2, 0.5 and 0.2, respectively for the three agents, and it is apparent from the figure that the change in Agent 2's strategy is the largest among the three investors, due to the largest competition weight factor.

Table \ref{tab:nl_mean_std_cv_ratio} shows the time series statistics of investment strategies for all investors with different market parameters and risk preferences under the nonlinear return dynamics. Both the mean and Std are mean of the $ B = 64 $ sample path Means and Stds. The CV is defined similar as before as the standard deviation per unit of the mean, and we report the mean of CVs for the three investors in the table. We then calculate the changes in the CVs and report the ratio as an indicator to the volatilities of the strategies. Under the nonlinear dynamics, the CVs for the first parameter set is significantly smaller in the partial information case (PI), compared to the full information case (FI). For other initial return rates and mean reverting levels, it is similar that the strategies under PI is less volatile. Unlike the experimented cases of the linear filter, competition can decrease the volatility of the strategies, since the ratio of CVs for C and NC for the last set of parameters, $ (h_0, \bar{\mu}) = (0.1, 0.02) $ is smaller than $ 1 $. 

Table \ref{tab:nl_sharpe_vrr} reports the empirical Sharpe ratio and VRR, which is defined as the mean return over variance. The Sharpe ratio is higher in the full information case for the first two sets of market parameters, while for the last parameter set, the partial information case yields a higher Sharpe ratio. Similarly for the VRRs. The standard deviation of wealth processes may be large in the case of full information to the level that a high return could not compensate for, and thus induces a smaller empirical Sharpe ratio in the FI case. Whether or not competition increases the return of portfolio per unit of risk also may depend on the specification of risk preference and market parameters.

\begin{figure}[t]
	\centering
	\includegraphics[width=0.8\linewidth]{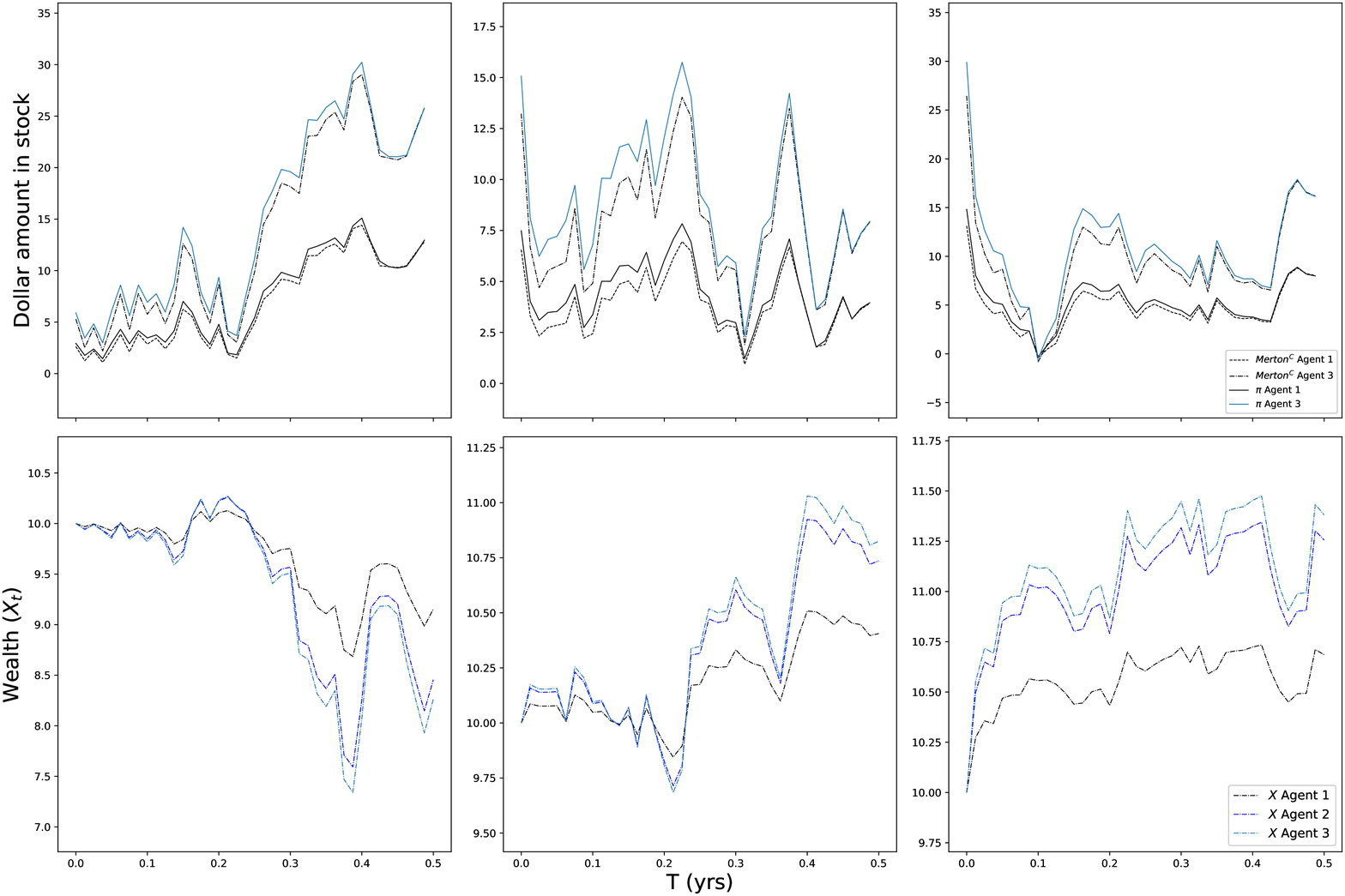}
	\caption{Sample paths of the dollar amount invested in the stock for Agent $ 1 $ and Agent $ 3 $, as well as the wealth processes $ X_t $ for all three agents when the return rate dynamics is nonlinear. The left column corresponds to investors' response when the market parameters are $ h_0 = 0.02 $, $ \bar{\mu} = 0 $. The middle column corresponds to the investors' response when the market has $ h_0 = 0.05 $, $ \bar{\mu} = 0.02 $. And the right column corresponds to the strategies and wealth processes for investors' when the true market parameters are $ h_0 = 0.1 $ with $ \bar{\mu} = 0.02 $.}
	\label{fig:pixnlhomo}
\end{figure}

\begin{figure}[t]
	\centering
	\includegraphics[width=0.8\textwidth]{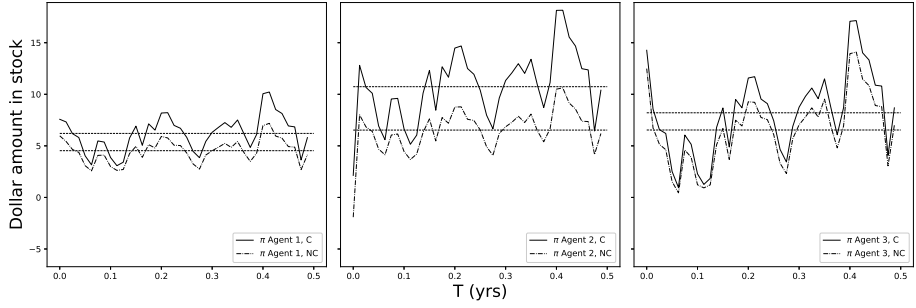}
	\caption{Sample paths of investors' strategies under nonlinear return rate dynamics. From the left to the right are the strategies for Agent $ 1 $, Agent $ 2 $ and Agent $ 3 $, respectively. The solid lines are strategies when investors are under competition. The dash-dotted lines are the Merton strategies. The risk aversion parameters and competition weights for both cases are specified in \ref{tab:risk_params}. The horizontal lines are the averages of strategies across time for the plotted sample paths.}
	\label{fig:picncnl}
\end{figure}

\begin{table}[t]
	\resizebox{0.95\textwidth}{!}
	{\begin{tabular}{@{}llllllllll@{}}
			\toprule
			 & \multicolumn{3}{c}{Mean of (abs. value) strategies} & \multicolumn{3}{c}{Std of (abs. value) strategies} & CV & \multicolumn{2}{l}{Ratio of CV} \\
			 \midrule 
			& \multicolumn{1}{c}{\textbf{Agent 1}} & \multicolumn{1}{c}{\textbf{Agent 2}} & \multicolumn{1}{c}{\textbf{Agent 3}} & \multicolumn{1}{c}{\textbf{Agent 1}} & \multicolumn{1}{c}{\textbf{Agent 2}} & \multicolumn{1}{c}{\textbf{Agent 3}} & Mean & \multicolumn{1}{c}{\textbf{PI / FI}} & \multicolumn{1}{c}{\textbf{C / NC}} \\
			\textbf{NC-FI} & 5.386 & 8.078 & 13.461 & 3.671 & 5.494 & 9.163 & 0.681 &  &  \\
			\textbf{NC-PI} & 3.024 & 4.540 & 7.566 & 1.625 & 2.438 & 4.063 & 0.537 & 0.789 &  \\
			\textbf{C-PI} & 4.092 & 7.454 & 8.249 & 2.187 & 3.989 & 4.415 & 0.535 &  & 0.996 \\
			\midrule
			\textbf{NC-FI} & 6.064 & 9.099 & 15.168 & 3.740 & 5.612 & 9.359 & 0.617 &  &  \\
			\textbf{NC-PI} & 3.735 & 5.606 & 9.333 & 1.556 & 2.331 & 3.881 & 0.416 & 0.675 &  \\
			\textbf{C-PI} & 5.507 & 10.039 & 11.098 & 2.397 & 4.367 & 4.833 & 0.435 &  & 1.046 \\
			\midrule
			\textbf{NC-FI} & 6.865 & 10.293 & 17.146 & 4.022 & 6.027 & 10.044 & 0.586 &  &  \\
			\textbf{NC-PI} & 5.100 & 7.645 & 12.741 & 1.845 & 2.773 & 4.615 & 0.362 & 0.618 &  \\
			\textbf{C-PI} & 7.322 & 13.347 & 14.763 & 2.602 & 4.740 & 5.256 & 0.356 &  & 0.982 \\ \bottomrule
	\end{tabular}}
	\caption{Time series statistics of the three agents' absolute value of investment strategies under different market parameter sets for the nonlinear case. The return dynamic is (\eqref{nl_stock} - \eqref{nl_hidden}) and the parameter set is in \eqref{nl_params}. The CV (coefficient of variation) is the time series standard deviation per unit of the mean, i.e, CV = Std / mean. Mean of CVs is the average of investors' CVs. The ratio PI / FI is the ratio of the (mean of) CVs of the case PI-NC and the FI-NC. The ratio C / NC is the ratio of the (mean of) CVs of the case PI-C and PI-NC. The top, middle and bottom panels correspond to the case $ (h_0, \bar{\mu}) $ equals to $ (0.02, 0) $, $ (0.05, 0.02) $ and $ (0.1, 0.02) $, respectively. The investor's initial belief is equal to the true return for all cases. Agents' risk parameters are the base parameters.}
	\label{tab:nl_mean_std_cv_ratio}
\end{table}

\begin{table}[t]
	\label{nl_sharpe_vrr}
	\resizebox{0.9\textwidth}{!}
	{\begin{tabular}{clllllllll}
	\toprule
	& \multicolumn{2}{l}{Sharpe ratio} &  &  & VRR &  &  &  &  \\
		\midrule
& \textbf{Agent 0} & \textbf{Agent 1} & \textbf{Agent 2} & \textbf{Social} & \textbf{Agent 0} & \textbf{Agent 1} & \textbf{Agent 2} & \textbf{Social} \\
\midrule
		\textbf{NC-FI} & 0.05613 & 0.05635 & 0.05642 & 0.01467 & 0.09359 & 0.06289 & 0.03782 & 0.004910 &  \\
		\textbf{NC-PI} & 0.009739 & 0.009683 & 0.009705 & 0.001139 & 0.02871 & 0.01899 & 0.01145 & 0.0006710 &  \\
		\textbf{C-PI} & 0.02875 & 0.02900 & 0.02873 & 0.006912 & 0.06703 & 0.03725 & 0.03315 & 0.003329 &  \\
		\hline 
		\textbf{NC-FI} & 0.05903 & 0.05926 & 0.05933 & 0.01523 & 0.08997 & 0.06037 & 0.03626 & 0.004654 &  \\
		\textbf{NC-PI} & 0.03445 & 0.03445 & 0.03440 & 0.009266 & 0.09866 & 0.06578 & 0.03945 & 0.005312 &  \\
		\textbf{C-PI} & 0.01656 & 0.01658 & 0.01647 & 0.00113 & 0.02605 & 0.01432 & 0.01286 & 0.0003680 &  \\
		\hline
		\textbf{NC-FI} & 0.01569 & 0.01553 & 0.01557 & -0.003970 & 0.01729 & 0.01139 & 0.006854 & -0.0008700 &  \\
		\textbf{NC-PI} & 0.03577 & 0.03571 & 0.03573 & 0.008800 & 0.07407 & 0.04918 & 0.02958 & 0.003641 &  \\
		\textbf{C-PI} & 0.07469 & 0.07465 & 0.07486 & 0.02095 & 0.1460 & 0.08064 & 0.07303 & 0.008520 &  \\
		\bottomrule
	\end{tabular}}
	\caption{The empirical Sharpe ratio and VRR for each agent as well as the social wealth for nonlinear return rate dynamics. The top panel (row 1 - 3) corresponds to the market parameters with initial return rate $ 0.02 $ and mean reverting to $ 0 $. The middle panel (row 4 - 6) corresponds to initial return $ 0.05 $ and mean reverting level $ 0.02 $. The bottom panel (row 7 - 9) is for different information setting when the initial market return is $ 0.1 $, with mean-reverting to $ 0.02 $. In each market setting, the investors have the initial belief that is a constant equals to the true market return rate. Agents' risk parameters are the base parameters.}
	\label{tab:nl_sharpe_vrr}
\end{table}

\subsection{Heterogenous prior beliefs: the nonlinear filter and comparisons}
We have described the computation algorithm of the estimation step, Stage I in previous sections. The estimations enter the FBSDE as state variables, and allow us to solve for the optimal investment strategies under both the case with and without competition. Table \ref{tab:hetero_pi_l_nl} shows the times series mean and standard derivations of the absolute value of investment strategies of the 3 agents. Similar as in the previous sections, we report the mean of the time series statistics over a sample of size $ B = 64 $. The bottom panel of \ref{tab:hetero_pi_l_nl} is the differences of the statistics between the linear and nonlinear case for the set of market parameters that we specified. 


The top panel shows that competition increases the investment proportion and its volatility, since the ratio of the mean of the competition and no competition case is greater than 1 for all agents both in the homogeneous and heterogeneous case. Similarly for the case with nonlinear filters as it is shown in the middle panel. A key observation is the ratios are larger in the HT case compared to the HM case for both linear and nonlinear return dynamics, which indicates that investors increase their investment proportion in absolute value. In other words, Agent 1 and Agent 2 follow the strategy of Agent 3, who has the most aggressive strategy. Agent 3 also increases the investment proportion due to the aggregate effect of strategies of the other two and the competition effect. Since Agent 3 has the most accurate prior, this further illustrates that the investor with the most accurate information leads the investment strategy when market agents interact through the relative wealth concern.

The bottom panel is the difference between the top and the middle panel, the case L minus the case NL. The ratio is higher in the linear filter case, except for Agent 3 with heterogeneous investors, where the competition effect is more pronounced for Agent 3 in the nonlinear case compared to the linear case. Viewing the competition effect as an agent-market characteristics, the cross-market difference varies more across agents in the HT case (the column C / NC in the bottom panel). Therefore, partial information heterogeneity affects the sensitivity of the competition effect to market characteristics, the return process in this particular case. 

Notice that the information heterogeneity is only on the prior estimates. All agents follow the Bayesian learning procedure with unlimited information processing ability. The effect of heterogeneity is already pronounced.

\begin{table}[t]
\resizebox{0.9\textwidth}{!}
{\begin{tabular}{@{}lclllllllll@{}}
		\toprule
		& & \multicolumn{3}{l}{Mean of (abs. value) strategies} & \multicolumn{3}{l}{Std of (abs. value) strategies} & \multicolumn{1}{c}{\textbf{C / NC}} &  &  \\
		\hline 
		& \multicolumn{1}{l}{} & \multicolumn{1}{c}{\textbf{Agent 1}} & \multicolumn{1}{c}{\textbf{Agent 2}} & \multicolumn{1}{c}{\textbf{Agent 3}} & \multicolumn{1}{c}{\textbf{Agent 1}} & \multicolumn{1}{c}{\textbf{Agent 2}} & \multicolumn{1}{c}{\textbf{Agent 3}} & \multicolumn{1}{c}{\textbf{Agent 1}} & \multicolumn{1}{c}{\textbf{Agent 2}} & \multicolumn{1}{c}{\textbf{Agent 3}} \\
		\midrule
		L-HT & \textbf{NC-PI} & 4.262 & 6.689 & 10.581 & 3.661 & 6.191 & 8.953 &  &  &  \\
		& \textbf{C-PI} & 6.654 & 12.399 & 13.350 & 5.148 & 10.000 & 10.131 & 1.561 & 1.854 & 1.262 \\
		L-HM & \textbf{NC-PI} & 4.315 & 6.472 & 10.786 & 2.516 & 3.777 & 6.291 &  &  &  \\
		& \textbf{C-PI} & 6.436 & 11.723 & 12.983 & 3.720 & 6.775 & 7.497 & 1.492 & 1.811 & 1.204 \\
		\hline 
		NL-HT & \textbf{NC-PI} & 4.998 & 7.576 & 8.920 & 2.095 & 3.438 & 4.915 &  &  &  \\
		& \textbf{C-PI} & 7.402 & 13.426 & 12.043 & 2.992 & 6.084 & 5.991 & 1.481 & 1.772 & 1.350 \\
		NL-HM & \textbf{NC-PI} & 3.735 & 5.606 & 9.333 & 1.556 & 2.331 & 3.881 &  &  &  \\
		& \textbf{C-PI} & 5.507 & 10.039 & 11.098 & 2.397 & 4.367 & 4.833 & 1.474 & 1.791 & 1.189 \\
		\hline 
		Diff: HT & \textbf{NC-PI} & -0.735 & -0.888 & 1.661 & 1.566 & 2.753 & 4.038 &  &  &  \\
		& \textbf{C-PI} & -0.748 & -1.027 & 1.308 & 2.156 & 3.916 & 4.140 & 0.080 & 0.082 & -0.088 \\
		Diff: HM & \textbf{NC-PI} & 0.580 & 0.866 & 1.453 & 0.960 & 1.445 & 2.410 &  &  &  \\
		& \textbf{C-PI} & 0.929 & 1.684 & 1.885 & 1.323 & 2.407 & 2.664 & 0.017 & 0.020 & 0.015 \\
		\bottomrule
\end{tabular}}
\caption{The mean and standard deviation of absolute value of investment strategies for investors with hetegogeneous initial beliefs, under both the linear and nonlinear return dynamics. L stands for linear filter, NL for nonlinear filter, HM and HT stands for homogeneous and nonhomogeneous agents, respectively. In each section of case L and NL, the first two rows show strategies for the heterogeneous beliefs, while we include the result from the homogeneous investors case for comparison (the last two rows in each section). The base market parameter is $ (h_0, \bar{\mu}) = (0.05, 0.02) $. The heterogeneous investors have the initial estimation of $ \hat{h}_0 = h_0 $, while differ in the variance parameters. For the linear case, the variances are $ 0.05 $, $ 0.1 $ and $ 0 $, respectively for the 3 agents. For the nonlinear case, the initial belief is drawn from a uniform distribution on a bounded interval that is with the length $ 0.05 $, $ 0.1 $ and $ 0 $, respectively and centered at the mean $ h_0 $. The bottom panel is the difference of the Means and Stds of the L and NL cases from the top and middle panel.}
\label{tab:hetero_pi_l_nl}
\end{table}




\clearpage
\pagebreak

\section{Conclusion and further remarks}
In this paper, we consider an N-agent game where the investors are utility maximizers. The investor's utility function depends on the wealth amount she outperforms the market average. Market investors can only observe the stock prices, but not the state that drives the drift. First, we establish a fully-coupled forward backward stochastic differential equation (FBSDE) that characterizes the N-agent investment decisions. For bounded return process, we show that the FBSDE solution is unique. Therefore, for the linear Gaussian or bounded nonlinear returns , we have the existence and uniqueness result of the FBSDE solution. Hence, there is a unique Nash equilibrium for the game. The wellposedness of the FBSDE in the case of unbounded return process is not readily available, because it requires higher moments estimation of the solution components to meet the coupling condition. We leave it for future reseach. For the numerical scheme, we apply a novel deep learning approach to the system of equations. We first apply deep-neural-network-based $ L^2 $ projection to obtain each investor's estimation of the asset return, and then design a deep FBSDE solver to find the value functions and the optimal controls simultaneously for all agents. The deep learning solution is compared to the PDE solution for the linear Gaussian return. The methodology developed in this paper, both the theoretical results and the numerical methods have potential applications in stochastic controls, stochastic games as well as in the mean field setting. Moreover, in the present paper, the information heterogeneity is only on the prior estimates. All agents follow the Bayesian learning procedure with uncounstrained information processing ability. The effect of heterogeneity is already pronounced. The case with agents' heterogeneity in information capacity and the corresponding asset pricing implications are promising reseach directions that will potentially lead to fruitful insights. 

\vspace{0.8em}

\section*{Acknowledgements}
We are grateful for the financial support from the NSF of China (Nos.11801099 and 11871364). Chao Deng also appreciates the financial support from the Ministry of education of Humanities and Social Science project of China (No. 18YJC910005), and the Natural Science Foundation of Guangdong Province (No. 2017A030310575). Chao Zhou’s work is also supported by Singapore MOE (Ministry of Educations) AcRF Grants R-146-000-219-112, R-146-000-255-114, and R-146-000-271-112 as well as the French Ministry of Foreign Affairs and the Merlion programme. Xizhi Su acknowledges the financial support from Centre for Quantitative Finance at NUS.

\section*{Appendices}
\subsection{Derivation of PDE solutions.}
When the return process is a linear function of $ (A_t)_{t \in [0, T]} $, we follow the market model \eqref{ou_stock} and \eqref{ou_hidden} and solve for the utility maximization problem using PDE approach. The optimal control and value function can be characterized by HJB equation. When a unique classical solution exists for the PDE, we can apply the Ito's formula to verify the solution is the value function of the control problem. The optimal control is obtained as a byproduct. We now state the verification theorem for classical solutions.

Let $ w $ be a function in $ C^{1, 2}([0, T] \times \mathbb{R}^N) $ solution to the HJB equation:
\begin{align*}
\frac{\partial w}{\partial t}(t, x) + \underset{a \in \mathcal{A}}{\sup}\pr{\mathcal{L}^aw(t, x) + f(x, a)} &= 0, \quad (t, x) \in [0, T) \times \mathbb{R}^N, \\
w(T, x) &= g(x), \quad x \in \mathbb{R}^N.
\end{align*}

\textbf{Verification theorem.} Suppose there exists a measurable function $ \hat{a}(t, x), \; (t, x) \in [0, T] \times \mathbb{R}^N $, valued in $ \mathcal{A} $ attaining the supremum, i.e.
\begin{equation*}
\underset{a \in \mathcal{A}}{\sup} \pr{\mathcal{L}^a w(t, x) + f(x, a)} = \mathcal{L}^{\hat{a}(t, x)} w(t,x) + f(x, \hat{a}(t, x)),
\end{equation*}
such that the SDE
\begin{equation*}
dX_u = b(X_u, \hat{a}(u, X_u)) du + \sigma(X_u, \hat{a}(u, X_u)) dW_u
\end{equation*}
admits a unique solution denoted by $ \hat{X}_u^{t, x} $, $ t \leq u \leq T $, with the initial condition $ X_t = x $, and the process $ \hat{\alpha} = \{\hat{a}(u, \hat{X}_u^{t, x}), \: t \leq u \leq T \}$ lies in $ \mathcal{A} $, then $ w = v $, and $ \hat{\alpha} $ is an optimal feedback control.

\textbf{Partial information HJB equation.} For CARA utility, the solution $ V(t, x, y, \eta) $ is smooth. Hence the classical verification theorem applies, which states that if the HJB equation has a smooth solution, then the solution is the value function to the control problem. Let 
\[ V(t, x, y, \eta) = \underset{\pi \in \mathcal{A}}{\sup} \mathbb{E} \pr{J(t) \big | X_t = x, \tilde{X}_t = y, \hat{\mu}_t = \eta, \hat{\Sigma}(t) = \sigma_0} . \]
Omitting the script $ i $ when there is no ambiguity. The agent $ i $'s value function $ V(t, \eta, x, y) $ satistifes the HJB equation, 
\begin{multline*}
V_t + \underset{\pi_t}{\sup} \bigg\{\pi_t \eta V_x + \alpha_t^{-i} \eta V_y - \lambda (\eta - \bar{\mu}) V_\eta + \frac{1}{2} \pi_t^2 \sigma_S^2 V_{xx} +  \frac{1}{2} (\alpha_t^{-i})^2 \sigma_S^2 V_{yy} \\ + \frac{1}{2} \br{\frac{\hat{\Sigma}(t) + \sigma_S \sigma_\mu \rho}{\sigma_S}}^2 V_{\eta \eta} 
+ \alpha_t^{-i} \pi_t \sigma_S^2 V_{xy} + \pi_t \br{\hat{\Sigma}(t) + \sigma_S \sigma_\mu \rho} V_{x \eta}  \\ + \alpha_t^{-i} \br{\hat{\Sigma}(t) + \sigma_S \sigma_\mu \rho} V_{y \eta} \bigg\}= 0,
\end{multline*}
with terminal condition $ V(T, x, y, \eta) = - e^{-\frac{1}{\delta}((1 - \frac{\theta}{N})x - \theta y)} $.

By the first order condition, the optimal $ \pi_t $ is 
\begin{equation*}
\pi_t^* = - \frac{\eta V_x + \alpha^{-i}_t \sigma_S^2 V_{xy} + \br{\hat{\Sigma}(t) + \sigma_S \sigma_\mu \rho} V_{x \eta}}{\sigma_S^2 V_{xx}}.
\end{equation*}

Substitute $ \pi_t^* $ into the above equation, we obtain the PDE for value function,
\begin{multline*}
V_t + \alpha_t^{-i} \eta V_y - \lambda (\eta - \bar{\mu}) V_\eta + \frac{1}{2} (\alpha_t^{-i})^2 \sigma_S^2 V_{yy} + \frac{1}{2} \br{\frac{\hat{\Sigma}(t) + \sigma_S \sigma_\mu \rho}{\sigma_S}}^2 V_{\eta \eta} \\ - \frac{\br{\eta V_x + \alpha_t^{-i} \sigma_S^2 V_{xy} + (\hat{\Sigma}(t) + \sigma_S \sigma_\mu\rho) V_{x \eta}}^2}{2 V_{xx} \sigma_S^2} = 0,
\end{multline*}
with terminal condition $ V(T, x, y, \eta) = - e^{-\frac{1}{\delta}\br{\br{1 - \frac{\theta}{N}} x - \theta y}} $.

Make an ansatz $ V(t, x, y, \eta) = - e^{-\frac{1}{\delta}\br{\br{1 - \frac{\theta}{N}} x - \theta y}} f(t, \eta)$. The PDE for $ f(t, \eta) $ is given by 
\begin{multline*}
f_t + w_2 \alpha_t^{-i} \eta f - \lambda \br{\eta - \bar{\mu}} f_\eta + \frac{1}{2} w_2^2 (\alpha_t^{-i})^2 \sigma_S^2 f + \frac{1}{2} \br{\frac{\hat{\Sigma}(t) + \sigma_S \sigma_\mu \rho}{\sigma_S}}^2 f_{\eta \eta} \\
- \frac{\br{f + w_2 \alpha_t^{-i} \sigma_S^2 f + \br{\hat{\Sigma}(t) + \sigma_S \sigma_\mu\rho} f_{\eta}}^2}{2 \sigma_S^2 f} = 0
\end{multline*}
with $ f(T, \eta) = 1 $ and $ w_2 = \frac{\theta}{\delta} $. 

Further simplification gives
\begin{multline*}
f_t + \br{w_2 \alpha_t^{-i} \eta + \frac{1}{2} w_2^2 (\alpha_t^{-i})^2 \sigma_S^2 - \frac{\br{\eta + w_2 \alpha_t^{-i} \sigma_S^2}^2}{2 \sigma_S^2}} f + \frac{1}{2} \br{\frac{\hat{\Sigma}(t) + \sigma_S \sigma_\mu \rho}{\sigma_S}}^2 f_{\eta \eta} - \lambda  (\eta - \bar{\mu}) f_\eta \\ - \frac{\br{\eta + w_2 \alpha_t^{-i} \sigma_S^2}\br{\hat{\Sigma}(t) + \sigma_S \sigma_\mu\rho}}{\sigma_S^2} f_\eta - \frac{1}{2} \br{\frac{\hat{\Sigma}(t) + \sigma_S \sigma_\mu \rho}{\sigma_S}}^2 \frac{f_\eta^2}{f} = 0, 
\end{multline*}
with $ f(T) = 1 $. 

Set a transformation $ f(t, \eta) = e^{g(t, \eta)} $, then $ g(t, \eta) $ satisfies the PDE, 
\begin{multline*}
g_t + \frac{1}{2} \br{\frac{\hat{\Sigma}(t) + \sigma_S \sigma_\mu \rho}{\sigma_S}}^2 g_{\eta \eta} - \lambda (\eta - \bar{\mu}) g_\eta - \frac{\br{\eta + w_2 \alpha_t^{-i} \sigma_S^2}\br{\hat{\Sigma}(t) + \sigma_S \sigma_\mu\rho}}{\sigma_S^2} g_\eta \\ + w_2 \alpha_t^{-i} \eta + \frac{1}{2} w_2^2 (\alpha_t^{-i})^2 \sigma_S^2 - \frac{\br{\eta + w_2 \alpha_t^{-i} \sigma_S^2}^2}{2 \sigma_S^2} = 0,
\end{multline*}
with $ g(T, \eta) = 0 $.

For the above PDE that is second order in the variable $ \eta $, we make an ansatz that the solution is quadratic in $ \eta $ with coefficients as an integral with respect to $ t $: 
\begin{equation*}
g(t, \eta) = \int_{t}^{T}\pr{ A(t, s)\eta^2 + B(t, s) \eta + C(t, s)} ds
\end{equation*}
where $ A(t, s), B(t, s) $ and $ C(t, s) $ satisfy the following ODEs:

\begin{align*}
& \dot{A} - 2 \lambda A - \frac{2\br{\hat{\Sigma}(t) + \sigma_S \sigma_\mu \rho}}{\sigma_S} A = 0, \\
& \dot{B} - \lambda B - \frac{\hat{\Sigma}(t) + \sigma_S \sigma_\mu \rho}{\sigma_S^2} B + 2 \lambda \bar{\mu} A + 2 w_2 \alpha_t^{-i} \br{\hat{\Sigma}(t) + \sigma_S \sigma_\mu\rho} A = 0 ,  \\
& \dot{C} + \frac{\br{\Sigma(t) + \sigma_S \sigma_\mu\rho }^2}{\sigma_S^2} A + \lambda \bar{\mu} B + w_2 \alpha_t^{-i} \br{\hat{\Sigma}(t) + \sigma_S \sigma_\mu\rho} B = 0,
\end{align*}
with 
$A(t, t) = - \frac{1}{2 \sigma_S^2}, B(t, t) = C(t, t) = 0$. 

The solution to the ODE system is 
\begin{align}\label{A}
A(t, s) & = - \frac{1}{2 \sigma_S^2} e^{- 2 \int_{t}^{s} \br{\lambda + \frac{\hat{\Sigma}(u) + \sigma_S \sigma_\mu\rho}{\sigma_S} } du}, \\
\label{B}
B(t, s) & = l(t, s) e^{- \int_{t}^{s} \br{\lambda + \frac{\hat{\Sigma}(u) + \sigma_S \sigma_\mu\rho}{\sigma_S^2}} du}, 
\end{align}
for 
\begin{align}\label{l}
l(t, s) & = - \frac{1}{\sigma_S^2} \int_{t}^{s} \br{\lambda \bar{\mu} + w_2 \alpha_t^{-i} \br{\hat{\Sigma}(u) + \sigma_S \sigma_\mu\rho}} e^{- \int_{u}^{s} \br{\lambda + \frac{2\br{\hat{\Sigma}(m) + \sigma_S \sigma_\mu\rho}}{\sigma_S} - \frac{\hat{\Sigma}(m) + \sigma_S \sigma_\mu\rho}{\sigma_S^2} }dm} du,
\end{align}

and
\begin{align}
C(t, s) & = \int_{t}^{s} \br{\frac{(\hat{\Sigma}(u) + \sigma_S \sigma_\mu\rho)^2}{\sigma_S^2} A(u, s) + \br{\lambda \bar{\mu} + w_2 \alpha_t^{-i} (\hat{\Sigma}(u) + \sigma_S \sigma_\mu\rho)} B(u, s)} du.
\end{align}

Moreover, the strategy $ \pi^* $ is given by 
\begin{align*}\label{eqpi}
\pi_t^* & = \frac{\delta \eta + \br{\hat{\Sigma}(t) + \sigma_S \sigma_\mu\rho} \delta g_\eta(t, T) + \alpha_t^{-i} \sigma_S^2 \theta}{\sigma_S^2 (1 - \frac{\theta}{N})} \\\nonumber
& = \frac{\delta}{\sigma_S^2 (1 - \frac{\theta}{N})} \br{\eta + \br{\hat{\Sigma}(t) + \sigma_S \sigma_\mu\rho} \int_{t}^{T}\br{2 A(t, s) \eta + B(t, s)} ds} + \frac{\theta}{1 - \frac{\theta}{N}} \alpha^{-i}_t.
\end{align*}

More explicitly, 
\begin{multline*}
\pi^* = \frac{\delta}{\sigma_S^2 \br{1 - \frac{\theta}{N}}}\bigg\{\eta + \br{\hat{\Sigma}(t) + \sigma_S \sigma_\mu\rho} \int_{t}^{T} - \frac{1}{\sigma_S^2} e^{- 2 \int_{t}^{s} \br{\lambda + \frac{\hat{\Sigma}(u) + \sigma_S \sigma_\mu\rho}{\sigma_S}} du} \eta \\ + \frac{\hat{\Sigma}(u) + \sigma_S \sigma_\mu\rho}{\sigma_S^2} \int_{t}^{T} \int_{t}^{s} \lambda \bar{\mu} e^{- \int_{u}^{s} \br{\lambda + \frac{2(\hat{\Sigma}(m) + \sigma_S \sigma_\mu\rho)}{\sigma_S} - \frac{\hat{\Sigma}(m) + \sigma_S \sigma_\mu\rho}{\sigma_S^2}} dm} du  \\ e^{- \int_{t}^{s} \br{\lambda + \frac{\hat{\Sigma}(u) + \sigma_S \sigma_\mu\rho}{\sigma_S^2}} du} ds \bigg\}
+ \frac{\delta}{\sigma_S^2 (1 - \frac{\theta}{N})} \Bigg\{ \frac{ w_2^i \br{\hat{\Sigma}(t) + \sigma_S \sigma_\mu\rho} }{\sigma_S^2} \int_{t}^{T} e^{- \int_{t}^{s} (\lambda + \frac{\hat{\Sigma}(u) + \sigma_S \sigma_\mu\rho}{\sigma_S^2} )du} \\
\cdot \int_{t}^{s}  \br{\hat{\Sigma}(u) + \sigma_S \sigma_\mu\rho} e^{- \int_{u}^{s} \br{\lambda + \frac{2(\hat{\Sigma}(m) + \sigma_S \sigma_\mu\rho)}{\sigma_S} - \frac{\hat{\Sigma}(m) + \sigma_S \sigma_\mu\rho}{\sigma_S^2}} dm} du ds + \frac{\theta}{1 - \frac{\theta}{N}}\Bigg\} \alpha^{-i}_t .
\end{multline*}

\vspace{0.5em}

\subsection{Appendix B}
The deep learning results with respect to different learning rates are shown in the figure below. 

\begin{figure}[tb]
	\centering
	\includegraphics[width=0.9\linewidth]{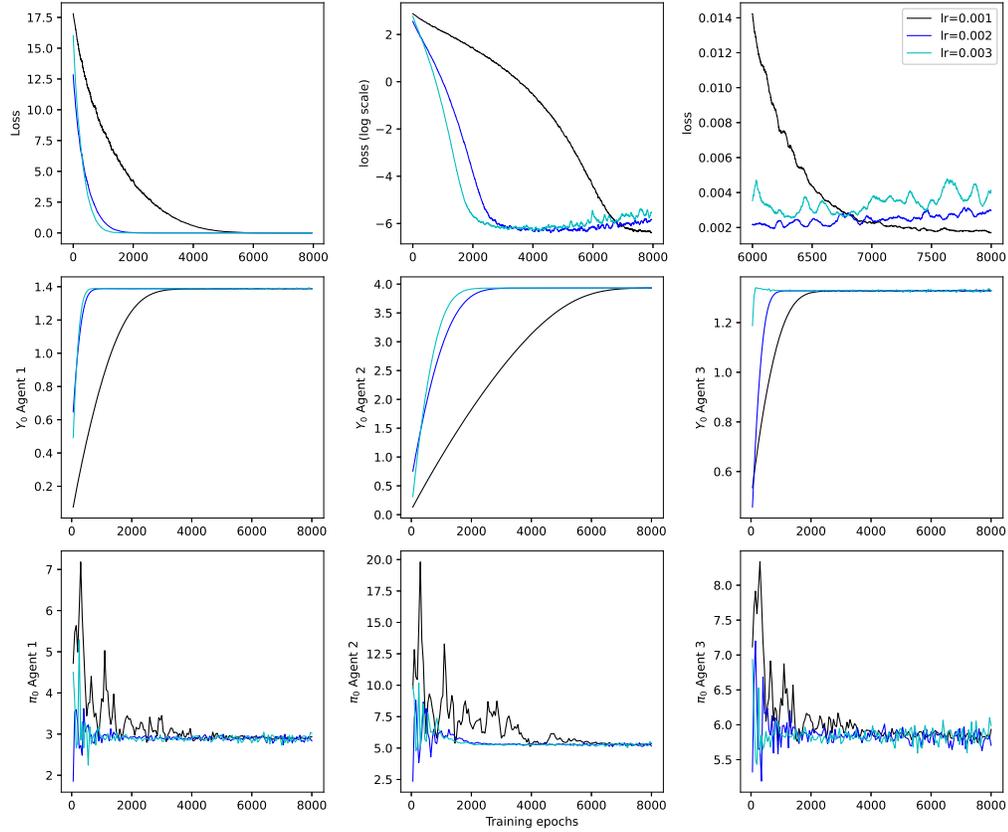}
	\caption{The convergence of FBSDE solutions with respect to the training epochs with different (constant without decay) learning rates. The top panel shows the loss quantity, with the left two figures showing the loss in the ordinary scale and the log scale, respectively. The right-most figure shows the loss for the last $2000$ training epochs. The mid row is the initial $ Y $ values corresponding to the components of the FBSDE solution in $ \mathbb{R}^N $. From the left to the right, it is the first, second and the third component, respectively. The last row shows the initial investment amount for agent $ 1 $, agent $ 2 $ and agent $ 3 $, respectively from the left to the right.}
	\label{fig:lr}
\end{figure}

\clearpage

\bibliographystyle{apa}
\bibliography{document}

\end{document}